\documentclass[11pt,letterpaper]{article}

\usepackage{geometry}
\usepackage[numbers]{natbib}

\usepackage{titlesec}

\usepackage{amsmath}
\usepackage{amsfonts}
\usepackage{amssymb}
\usepackage{amsthm}
\usepackage{amsbsy}
\usepackage{graphicx}

\usepackage{verbatim}
\usepackage{url}
\usepackage{bbm}
\usepackage{multirow}
\usepackage{enumitem}
\usepackage{xspace}
\usepackage{footnote}
\usepackage{hyperref}
\usepackage[svgnames]{xcolor}
\usepackage[capitalise,nameinlink]{cleveref}
\hypersetup{colorlinks={true},linkcolor={DarkBlue},citecolor=[named]{DarkGreen},urlcolor=black}
\usepackage{bm}
\usepackage{tikz}
\usepackage{algorithm}
\usepackage{algpseudocode}
\renewcommand{\algorithmiccomment}[1]{\bgroup\hfill\small\textcolor{gray}{//~#1}\egroup}
\algnotext{EndFor}
\algnotext{EndIf}
\algnotext{EndWhile}
\algnotext{EndProcedure}
\makeatletter
\renewcommand{\ALG@name}{Mechanism}
\makeatother

\usepackage{titling}
\usepackage[noblocks]{authblk}

\newtheorem{theorem}{Theorem}[section]
\newtheorem{lemma}[theorem]{Lemma}
\newtheorem{claim}[theorem]{Claim}

\newtheorem{corollary}[theorem]{Corollary}

\newtheorem*{claim*}{Claim}

\theoremstyle{definition}
\newtheorem{definition}[theorem]{Definition}
\newtheorem{remark}[theorem]{Remark}

\usepackage{libertine}

\usepackage{soul}
\setuldepth{p}

\newcommand{\bmu}{\ensuremath{\boldsymbol{\mu}}\xspace}
\newcommand{\mms}{\ensuremath{\textrm{\MakeUppercase{mms}}}\xspace}
\newcommand{\pmms}{\ensuremath{\textrm{\MakeUppercase{pmms}}}\xspace}
\newcommand{\gmms}{\ensuremath{\textrm{\MakeUppercase{gmms}}}\xspace}
\newcommand{\prop}{\ensuremath{\textrm{\MakeUppercase{prop}}}\xspace}
\newcommand{\efo}{\ensuremath{\textrm{\MakeUppercase{ef}1}}\xspace}
\newcommand{\efx}{\ensuremath{\textrm{\MakeUppercase{efx}}}\xspace}
\newcommand{\ef}{\ensuremath{\textrm{\MakeUppercase{ef}}}\xspace}

\newenvironment{myquote}[1]%
  {\list{}{\leftmargin=#1\rightmargin=#1}\item[]}%
  {\endlist}

\newcommand\circled[1]{%
	\tikz[baseline=(char.base)]\node[circle,draw,inner sep=2pt] (char) {#1};}

\DeclareMathOperator*{\argmin}{arg\,min}
\DeclareMathOperator*{\argmax}{arg\,max}

\allowdisplaybreaks

\usepackage[skip=.4\baselineskip plus 2pt]{parskip}

\begin{document}

\author[1,2]{Georgios Amanatidis}
\affil[1]{School of Mathematics, Statistics and Actuarial Science, University of Essex, UK.}
\affil[2]{Archimedes/Athena RC, Greece.}
\author[3]{Georgios Birmpas}
\affil[3]{Department of Computer Science, University of Liverpool, UK.}
\author[4]{Federico Fusco}
\affil[4]{Department of Computer, Control, and Management Engineering, Sapienza University of Rome, Italy. }
\author[5]{Philip Lazos}
\affil[5]{Input Output Global (IOG), London, UK. }
\author[4]{Stefano Leonardi}
\author[6]{Rebecca Reiffenh{\"a}user}
\affil[6]{Institute for Logic, Language and Computation, University of Amsterdam, The Netherlands.  }
\title{Allocating Indivisible Goods to Strategic Agents: \\ Pure Nash Equilibria and Fairness\thanks{\textcolor{blue}{The claim that every pure Nash equilibrium of Round-Robin is EF1 (Theorem \ref{thm:main_theorem_for_n} here; Theorem 3 in the conference version \citep{AmanatidisBFLLR21} and Theorem 4 in the journal version \citep{AmanatidisBFLLR24}) is incorrect; the gap is in the proof of Lemma \ref{lem:main_lemma} (inequality \eqref{eq:sum_of_lambdas}). A counterexample and corrected guarantees are given by \citet{AABFLR26}. The remaining results are not affected.}}}

\predate{}
\postdate{}
\date{}

	\maketitle
	
\begin{abstract}
\noindent We consider the problem of fairly allocating a set of indivisible goods to a set of \emph{strategic} agents with additive valuation functions. We assume no monetary transfers and, therefore, a \textit{mechanism} in our setting is an algorithm that takes as input the reported---rather than the true---values of the agents. Our main goal is to explore whether there exist mechanisms that have pure Nash equilibria for every instance and, at the same time, provide fairness guarantees for the allocations that correspond to these equilibria. We focus on two relaxations of envy-freeness, namely  \textit{envy-freeness up to one good} (\efo), and \textit{envy-freeness up to any good} (\efx), and we positively answer the above question. In particular, we study two algorithms that are known to produce such allocations in the non-strategic setting: Round-Robin (\efo allocations for any number of agents) and a cut-and-choose algorithm of Plaut and Roughgarden \cite{PR18} (\efx allocations for two agents). 
For Round-Robin we show that all of its pure Nash equilibria induce allocations that are \efo with respect to the underlying true values, while for the algorithm of Plaut and Roughgarden we show that the corresponding allocations not only are \efx but also satisfy \textit{maximin share fairness}, something that is not true for this algorithm in the non-strategic setting! Further, we show that a weaker version of the latter result holds for any mechanism for two agents that always has pure Nash equilibria which all induce \efx allocations.
\end{abstract}

	\section{Introduction}
	
	Fair division refers to the problem of distributing a set of resources among a set of agents in such a way that everyone is ``happy'' with the overall allocation. Capturing this ``happiness'' can be elusive,
	as it may be determined by complicated underlying social dynamics; however, two well-motivated (and mathematically conducive) interpretations are those of \emph{envy-freeness} \citep{GS58,Foley67,Varian74} and \emph{proportionality} \citep{Steinhaus49}. When an allocation is envy-free, each agent values the set of resources that she receives at least as much as the set of any other agent, while when an allocation is proportional, each agent receives at least $1/n$ of her total value for all the goods, assuming there are $n$ agents.
	Since the first mathematically formal treatment of fair division  by Banach, Knaster, and Steinhaus \citep{Steinhaus49}, the multifaceted questions that arise for the different variants of the problem have been studied in a diverse group of fields, including mathematics, economics, and political science. As many of these questions are inherently algorithmic, fair division questions, especially the ones related to the existence, computation, and approximation of different fairness notions, have been very actively studied by computer scientists during the last two decades (see, e.g., \citep{Procaccia16-survey,BCM16-survey,Markakis17-survey,AmanatidisABFLMVW23} for surveys of recent results).

    In the standard discrete fair division setting that we study here, the resources are indivisible goods and the agents have additive valuation functions over them. Typically, there is also the additional assumption that all the goods need to be allocated. This discrete setting poses a significant conceptual challenge, as the classic notions of fairness originally introduced for \emph{divisible} goods, such as envy-freeness and proportionality, are  impossible to satisfy. The example that illustrates this situation needs only two agents and just one positively valued good. Whoever does not receive the good will not consider the result to be either envy-free or proportional.  However, this should not necessarily be considered an \emph{unfair} outcome, as it is done out of necessity, not malice: the only other (deterministic) option would be to deprive both agents of the good, which seems wasteful. To define what is  \emph{fair} in this context, a number of weaker fairness notions have been proposed. Among the most prevalent of those are \textit{envy-freeness up to one good} (\efo), \textit{envy-freeness up to any good} (\efx), and \textit{maximin share fairness} (\mms). 
   The notions of \efo and \efx were introduced by Lipton et al.~\cite{LMMS04}, Budish \cite{Budish11}, and Gourvès et al.~\cite{GourvesMT14}, Caragiannis et al.~\cite{CaragiannisKMPS19} respectively, and they can be seen as additive relaxations of envy-freeness. Both of them are based on the following rationale: an agent may envy another agent but only by the value of the most (for \efo) or the least (for \efx) desirable good in the other agent's bundle. It is straightforward that \efo is weaker than \efx, and indeed this is reflected to the known results for the two notions.
   The concept of the \textit{maximin share} of an agent was introduced by Budish \cite{Budish11}
   as a relaxation to the proportionality benchmark. The corresponding fairness notion, \textit{maximin share fairness} (MMS), requires that each agent receives the maximum value that this agent would obtain if she was allowed to partition the goods into $n$ bundles and then keep the worst of these (see Section \ref{sec:prelims} for a more detailed description and a formal definition).

   From an algorithmic point of view,  
   there are many  results  regarding the existence and computation of these notions (see our Related Work).
   Here, however, we are interested in exploring the problem from a game theoretic perspective. In particular, we assume that the agents are \textit{strategic}, which means that it is possible for an agent to intentionally misreport her values for (some of) the goods to end up with a bundle of higher total value.
   We see this as a very natural direction, as it captures what may happen in practice in many real-life scenarios where fair division solutions can be applied, for instance, in a divorce settlement.
    It should be noted here that, in accordance to the existing literature on truthful allocation mechanisms \cite{EhlersK03, KlausM02, Papai00, Papai01, ABM16, ABCM17, CKKK09}, 
    we assume there are \textit{no monetary transfers}. Therefore, a \textit{mechanism} in our setting is just an algorithm that takes as input the, possibly misreported, values that the agents declare. The existence of \textit{truthful} mechanisms, i.e., mechanisms where no agent ever has an incentive to lie, was studied  in the same setting by Amanatidis et al.~\cite{ABCM17} who showed that, even for two agents, truthfulness and fairness are incompatible by providing impossibility results for \emph{every} non-trivial fairness notion.
    As a consequence, the next natural question to ask is:   \smallskip
   \begin{myquote}{0.53in}
     \emph{Is it possible to have non-truthful mechanisms that are guaranteed to have equilibria, with these equilibria always inducing \emph{fair} allocations?} 
   \end{myquote} \smallskip
   Thus, our main quest is to investigate whether there exist mechanisms that have \textit{pure Nash equilibria} for every instance and each allocation corresponding to an equilibrium provides fairness guarantees with respect to the \textit{true} valuation functions of the agents.
   The stability notion of a pure Nash equilibrium, on which we focus here, describes a state where each agent plays a deterministic strategy (namely, reports her value for each good) and no agent can attain higher value by deviating to a different strategy.

	\subsection{Our Contributions}	
	To the best of our knowledge, our work is the first to consider the above question.
	The results we provide are mostly positive, as we show that the class of mechanisms that are implementable in polynomial time, have pure Nash equilibria for every instance, and provide some fairness guarantee at the allocations they produce in their equilibria is non-empty. 
Specifically, in Section \ref{sec:RR}, we study a mechanism adaptation of the Round-Robin algorithm which is known to produce \efo allocations in the non-strategic setting \citep{CaragiannisKMPS19}. Also, under some mild assumptions which we show that can be lifted, Aziz et al.~\cite{GW17} showed that the Round-Robin mechanism always has pure Nash equilibria.
Further, in Section \ref{sec:PR}, we consider the stronger fairness notion of \efx. We focus on the case of two agents and study a mechanism adaptation of the algorithm of Plaut and Roughgarden \cite{PR18}, Mod-Cut\&Choose,  which is known to always produce \efx allocations in the non-strategic setting.
Our main contributions can be summarized as follows:

\begin{itemize}[labelindent=15pt,leftmargin=*,itemsep=3pt,topsep=5pt]
    \item  Round-Robin has pure Nash equilibria for every instance and these equilibria induce allocations that are always \efo with respect to the underlying true values (Theorems \ref{thm:main_theorem_for_n} and \ref{thm:general_bluff}). That is, Round-Robin retains its fairness properties at its equilibria, even when the input is given by strategic agents!
    To show this, we combine well-known properties of Round-Robin with a novel recursive construction of ``nicely structured'' bid profiles. 
    We consider this as the main technical result of our paper.
    \item Mod-Cut\&Choose has pure Nash equilibria for every instance with two agents and these equilibria induce allocations that are always \efx {\em and} \mms with respect to the underlying true values (Theorem \ref{thm:mms+efx_PNE}). Note that for the case of two agents \mms allocations are always \efx allocations, i.e., \mms fairness is stronger. It should be also noted that in the non-strategic setting, for any $\varepsilon>0$, there are instances where the output of Mod-Cut\&Choose is not a $(5/6 + \varepsilon)$-\mms allocation!
    \item We generalize a weaker version of Theorem \ref{thm:mms+efx_PNE}. All mechanisms that have pure Nash equilibria for every instance with two agents and these equilibria induce allocations that are always \efx provide stronger \mms guarantees in these allocations than generic \efx allocations do (Theorems \ref{thm:4EFX=MMS} and \ref{thm:NEFX_APMMS}). This shows a very interesting separation between the strategic and non-strategic settings. 
	\end{itemize}

\subsection {Further Related Work}
The non-strategic version of the problem of fairly allocating goods to additive agents has been studied extensively. We provide a summary of indicative results mostly for the notions that we  consider. In particular, \efo allocations always exist and can be computed in polynomial time \cite{LMMS04,Markakis17-survey,CaragiannisKMPS19}. 
For the stronger notion of \efx, the picture is not that clear. It is known that such allocations always exist when there are $2$ or $3$ agents \cite{CaragiannisKMPS19, GourvesMT14, ChaGM20}, and in the former case they can be efficiently computed using Mod-Cut\&Choose \cite{PR18}. The existence of complete \efx allocations for $4$ or more agents remains one of the most  intriguing open problems in fair division. 
There are, however, positive results for any number of agents if the valuation functions are restricted \cite{ABFHV21, Mahara23, GargM23}, if it is allowed to discard some of the goods \cite{CaragiannisGH19, ChaudhuryKMS21, ChaudhuryGMMM21, BergerCFF22}, or if one considers approximate \efx allocations  \citep{PR18,ANM2019}. 
Finally, regarding the notion of \mms, allocations that provide this guarantee always exist when there are only $2$ agents, although computing them is an NP-hard problem \citep{Woeginger97}. Even worse, for three or more agents, such allocations do not always exist \citep{KurokawaPW18}.
However, there are algorithms that run in polynomial time and produce constant factor approximation guarantees \citep{KurokawaPW18, AMNS17, BarmanK20, GHSSY21, GargMT19, GargT21}, with ${3}/{4}+{3}/{3836}$ being the current state of the art \cite{AkramiGarg23}.

The works of Caragiannis et al.~\cite{CKKK09}, and Amanatidis, Birmpas and coauthors \cite{ABM16, ABCM17}  are very relevant to ours in the sense that they all studied the exact same strategic discrete fair division setting. As we mentioned earlier, however, their focus was different as they were only interested in truthful mechanisms. Amanatidis et al.~\cite{ABCM17} provided strong impossibility results in this direction: for instances with two agents, no truthful mechanism can consistently produce \efo (and thus \efx) allocations when there are more than $4$ goods, while the best possible approximation with respect to \mms declines linearly with the number of goods. Given these negative results, truthful mechanism design has also been  studied under restricted valuation function classes \citep{HPPS,BabaioffEF21,BabaioffF22}. In a very recent work, Amanatidis et al.~\cite{AmanatidisBL0R23} show that our main result (Theorem \ref{thm:main_theorem_for_n}) qualitatively extends to approximate pure Nash equilibria, even for agents with submodular valuation functions. 
	
	Aziz, Goldberg and Walsh \cite{GW17} studied the existence of pure Nash equilibria of Round-Robin and  showed that when no agent values any two goods equally, there always exists a pure Nash equilibrium. In addition, they provided a linear time algorithm that computes the preference rankings (i.e., the orderings of the goods that correspond to the reported values) that leads to this equilibrium, thus giving a constructive solution. 
	Aziz et al.~\cite{AzizBLM17} showed that computing best responses for Round-Robin, and for \emph{sequential mechanisms} more generally, is NP-hard, fixing an error in the work of Bouveret and Lang \cite{BouveretL14} on the same topic.
	
	We conclude by pointing out that in contrast to the case of indivisible goods, the problem of fairly allocating a set of divisible goods to a set of strategic agents has  been repeatedly studied. For some indicative papers in this line of work, we refer the reader to \cite{ColeGG13,ChenLPP13,BranzeiCKP16,BeiCHTW17,BranzeiGM22} and references therein.

\section{Preliminaries}\label{sec:prelims}
We consider the problem of allocating a set of indivisible goods to a set of agents in a fair manner under the presence of incentives. 
For $a\in \mathbb{N}$ we use $[a]$ to denote the set $\{1, 2, \ldots, a\}$. 	
An instance to our problem is an ordered triple $(N, M, \mathbf{v})$, where $N = [n]$ is a set of $n$ agents, $M = \{g_1,\ldots, g_m\}$ is a set of $m$ goods, and $\mathbf{v} = (v_1,\ldots, v_n)$ is a vector of the agents' additive valuation functions. In particular, each agent $i$ has a non-negative value $v_{i}(\{g\})$ (or simply $v_{i}(g)$)
for each good $g\in M$, and for every $S, T \subseteq M$ with $S\cap T= \emptyset$ we have $v_i(S\cup T) = v_i(S)+ v_i(T)$. Equivalently, the value of an agent is simply the sum of the values of the goods that she got.
We assume there is no free disposal, which means that all the goods must be allocated. Thus, an allocation $(A_1,\ldots,A_n)$, where $A_i$ is the \emph{bundle} of agent $i$, is a partition of $M$.
It is often useful 
to refer to the order of preference an agent has over the goods. We say that a valuation function $v_i$ \emph{induces a preference ranking} $\succeq_i$ if $g\succeq_i g' \Leftrightarrow v_i(g) \ge v_i(g')$ for all $g,g'\in M$. We use $\succ_i$ if the corresponding preference ranking is \emph{strict}, i.e., when $g\succeq_i g' \,\wedge\, g'\succeq_i g \,\Rightarrow\, g=g'$, for all  $g,g'\in M$.

\subsection{Fairness Notions}\label{subsec:fairness}
There is a significant number of different notions one can use to determine which allocations are ``fair''. The most prominent such notions are \emph{envy-freeness} (\ef) \cite{GS58,Foley67,Varian74} and \emph{proportionality} (\prop) \cite{Steinhaus49}, and, in the discrete setting we study here, their relaxations, namely \emph{envy-freeness up to one good} (\efo) \cite{Budish11}, \emph{envy-freeness up to any good} (\efx) \cite{CaragiannisKMPS19}, and \emph{maximin share fairness} (\mms) \cite{Budish11}. Particularly for additive valuation functions, we have that $\ef \Rightarrow \efx \Rightarrow \efo$ and $\ef \Rightarrow \prop \Rightarrow \mms$, where $X \Rightarrow Y$ means that any allocation that satisfies fairness criterion $X$ always satisfies fairness criterion $Y$ as well. 

\smallskip\begin{definition}\label{def:EF-EFX}
+	An allocation $(A_1,\ldots,A_n)$ is 
	\begin{itemize}[labelindent=20pt,leftmargin=*,itemsep=3pt,topsep=5pt]
		\item \textit{envy-free} (\ef), if for every $i, j\in N$, $v_i(A_i) \geq v_i(A_j)$. \label{def:EF}
		\item \textit{envy-free up to one good} (\efo), if for every pair of agents $i, j\in N$, with $A_j\neq\emptyset$, there exists a good $g\in A_j$, such that
		$v_i(A_i) \geq  v_i(A_j\setminus \{g\})$. \label{def:efo}
		\item  \textit{envy-free up to any good} (\efx), if for every pair $i, j\in N$, with $A_j\neq\emptyset$ and every good $g\in A_j$ with $v_i(g)>0$, it holds that $v_i(A_i) \geq v_i(A_j\setminus \{g\})$. \label{def:EFX} 
	\end{itemize}
\end{definition}\smallskip

While these notions rely on comparisons among the agents, proportionality focuses on everyone receiving at least a $1/n$ fraction of the total value.

\smallskip\begin{definition}\label{def:prop}
	An allocation $(A_1,\ldots,A_n)$ is \textit{proportional} (\prop), if for every $i\in N$, $v_i(A_i) \geq {v_i(M)}/{n}$.
\end{definition}\smallskip

In the same direction, but adjusted for indivisible goods, a number of fairness notions have been based on the notion of \emph{maximin shares} \cite{Budish11}. 
Imagine that agent $i$ is asked to partition the goods into $n$ bundles, under the condition that she will receive the worst bundle among those. If the resources were divisible, then she would clearly split everything evenly into $n$ bundles of value ${v_i(M)}/{n}$ each, thus capturing the benchmark required for proportionality. However, now that the goods are indivisible, agent $i$ would like to create a partition maximizing the minimum value of a bundle. This value is her maximin share. 
\smallskip\begin{definition}\label{def:mmshare}
	Given  a subset  $S\subseteq M$ of goods, the $n$-\textit{maximin share} of agent $i$ with respect to $S$ is
	\[ \bmu_i(n, S) = \displaystyle\max_{\mathcal{A}\in\Pi_n(S)} \min_{A_j\in \mathcal{A}} v_i(A_j)\,,\]
	where $\Pi_n(S)$  is the set of all partitions of $S$ into $n$ bundles.
\end{definition}\smallskip
From the definition and the preceding discussion, we have that $n\cdot \bmu_i(n, S)\le v_i(S)$.
When $S=M$, we call $\bmu_i(n, M)$ the \textit{maximin share} of agent $i$ and denote it by $\bmu_i$ as long as it is clear what $n$ and $M$ are. 
\smallskip\begin{definition}
	\label{def:MMS}
	An allocation 
	$\mathcal{A} = (A_1,\ldots,A_n) $ is called an \textit{$\alpha$-maximin share fair} ($\alpha$-\mms) allocation if $v_i(A_i)\geq \alpha\cdot \bmu_i\,$, for every $i\in N$. When $\alpha=1$ we just say that $\mathcal{A}$ is an \mms allocation.
\end{definition}\smallskip 
Besides \mms, there exist other fairness criteria based on the notion of maximin shares, like  \emph{pairwise maximin share fairness} (\pmms) \cite{CaragiannisKMPS19} and \emph{groupwise maximin share fairness} (\gmms) \cite{BBMN18}. While we are not going into more details about them, it should be noted that $\pmms\Rightarrow \efx$ \cite{CaragiannisKMPS19} and that for $n=2$, \mms, \pmms, and \gmms coincide. In particular, we need the following result of Caragiannis et al.~\cite{CaragiannisKMPS19}.

\begin{theorem}[Follows from Theorem 4.6 of \cite{CaragiannisKMPS19}]\label{thm:mms_to_efx_n=2}
	For $\,n = 2$, any \mms allocation is also an \efx allocation. 
\end{theorem}

In addition to the implications mentioned so far, one can consider how the approximate versions of \efo, \efx and \mms relate to each other (see \cite{ABM18}). Here we need the following result about the worst case \mms guarantee of an \efx allocation for the case of two agents.

\begin{theorem}[Follows from Proposition 3.3 of \cite{ABM18}]\label{thm:efx_to_mms_n=2}
	For $\,n = 2$, any \efx allocation is also a $\frac{2}{3}$-\mms allocation. 
	This guarantee is tight, in the sense that for every $\delta>0$ there exists an \efx allocation that is not a $\big(\frac{2}{3} +\delta\big)$-\mms allocation, for any $m\ge 4$. 
\end{theorem}

\subsection{Mechanisms and Equilibria}\label{subsec:mecanisms}
We are interested in \emph{mechanisms} that produce allocations with  fairness guarantees. In our setting, where there are \textit{no payments}, an allocation mechanism $\mathcal{M}$ is essentially just an algorithm that takes its input from the agents and allocates all the goods to them. We use this distinction in terminology to highlight that this reported input may differ from the actual valuation functions. In particular, we assume that each agent $i$ reports a \emph{bid vector} $\bm{b}_i = (b_{i1}, b_{i2}, \ldots, b_{im})$, where $b_{ij}\ge 0$ is the value agent $i$ claims to have for good $g_j\in M$. A mechanism $\mathcal{M}$ takes as input a \emph{bid profile} $\mathbf{b} = (\bm{b}_1, \bm{b}_2, \ldots, \bm{b}_n)$ of bid vectors and outputs an allocation $\mathcal{M}(\mathbf{b})$. 
In our setting we assume that the agents are \emph{strategic}, i.e., an agent may misreport her true values if this results to a better allocation from her point of view. Hence, in general, $\bm{b}_i \neq (v_{i}(g_1), v_{i}(g_2), \ldots, v_{i}(g_m))$. While $\bm{b}_i$ is defined as a vector, for a generic good $h\in M$ it is often convenient to use the function notation $\bm{b}_i(h)$ to denote the bid value ${b}_{i\ell}$, where $\ell$ is such that $h=g_{\ell}$; extending this we may write $\bm{b}_i(S)$ for $\sum_{h\in S} \bm{b}_i(h)$.
Like above, we say that a bid vector $\bm{b}_i$ induces a preference ranking $\succeq_i$ if $g\succeq_i g' \Leftrightarrow \bm{b}_i(g) \ge \bm{b}_i(g')$ for all $g,g'\in M$, and use $\succ_i$ for strict rankings.

We focus on the fairness guarantees of the (pure) equilibria of the mechanisms we study. As is common, given a profile $\mathbf{b} = (\bm{b}_1, \ldots, \bm{b}_n)$, we write
$\mathbf{b}_{-i}$ to denote $(\bm{b}_1, \ldots, \bm{b}_{i-1},  \allowbreak \bm{b}_{i+1}, \ldots, \bm{b}_n)$ and, given a bid vector $\bm{b}'_{i}$, we use 
$(\bm{b}'_i, \mathbf{b}_{-i})$ to denote the profile $(\bm{b}_1, \ldots, \bm{b}_{i-1}, \allowbreak \bm{b}'_{i}, \allowbreak \bm{b}_{i+1}, \ldots, \bm{b}_n)$. For the next definition we  abuse the notation slightly: given an allocation $\mathcal{A} = (A_1,\ldots,A_n)$, we write $v_i(\mathcal{A})$ to denote $v_i({A}_i)$. 

\smallskip\begin{definition}\label{def:PNE}
Let $\mathcal{M}$ be an allocation mechanism and consider a profile
$\mathbf{b} = (\bm{b}_1, \ldots, \bm{b}_n)$. We say that $\bm{b}_{i}$ is a \emph{best response} to $\mathbf{b}_{-i}$ if for every $\bm{b}'_{i}\in \mathbb{R}^m_{\ge 0}$, we have 
\[v_i(\mathcal{M}(\bm{b}'_i, \mathbf{b}_{-i}))\le v_i(\mathcal{M}(\mathbf{b})) \,.\]
The profile $\mathbf{b}$ is a \emph{pure Nash equilibrium} (PNE) if, for each $i\in N$, $\bm{b}_{i}$ is a best response to $\mathbf{b}_{-i}$.
\end{definition}\medskip

When $\mathbf{b}$ is a PNE and the allocation $\mathcal{M}(\mathbf{b})$ has a fairness guarantee, e.g., $\mathcal{M}(\mathbf{b})$ is \efo, we will atribute the same guarantee to the profile itself, i.e., we will say that $\mathbf{b}$ is \efo.

\smallskip\begin{remark}
The mechanisms we consider in this work run in polynomial time. However there are computational complexity questions that go beyond the mechanisms themselves. For instance,  how does an agent compute a best response or how do all the agents reach an equilibrium? 
While we consider such questions interesting directions for future work, we do not study them here and we only focus on the fairness properties of PNE. It should be noted, however, that such problems are typically hard. For instance, computing a best response for Round-Robin is NP-hard in general \citep{AzizBLM17} (although for fixed $n$ it can be done in polynomial time \citep{XiaoL20}), and the same can be easily shown to be true for
Mod-Cut\&Choose via a reduction from the classic PARTITION problem.
\end{remark}\smallskip

\begin{remark}\label{rem:MMS}
An easy observation on the main question of this work is that \emph{any} PNE of \emph{any}  $\alpha$-approxi\-ma\-tion mechanism for computing \mms allocations is an $\alpha$-\mms allocation. Indeed, this is true, not only for \mms  but for any fairness notion that depends on  agents achieving specific value benchmarks that depend on their own valuation function, e.g.,  it is also true for \prop. While this  is definitely interesting to note, nothing is known on the existence of PNE of any constant factor approxi\-ma\-tion algorithm for computing \mms allocations in the literature. Even for a very simple $1/2$-approximation algorithm that only slightly differs from Round-Robin  \cite{AMNS17}, showing that PNE always exist seems  very challenging.  Clearly, an existence result for any such algorithm \citep{KurokawaPW18,AMNS17,BarmanK20,GHSSY21,GargMT19,GargT21} would imply an analogue of Theorem \ref{thm:main_theorem_for_n} for appro\-xi\-mate \mms. 
Although in this work we do not consider \emph{mixed} Nash equilibria (MNE), i.e., the generalization of PNE where strategies are distributions over bids and the inequality of Definition \ref{def:PNE} holds in expectation, everything said in this remark could be repeated for MNE and ex-ante $\alpha$-\mms allocations, i.e., allocations where the inequality of Definition \ref{def:MMS} holds in expectation. We see all such questions as promising directions in line with the research agenda we initiate here.
\end{remark}\smallskip

\section{Fairness of Nash Equilibria of Round-Robin}\label{sec:RR}

In this section we focus on one of the simplest and most well-studied allocation algorithms, Round-Robin, a draft algorithm where the agents take turns and in each turn the active agent receives her most preferred available (i.e., unallocated) good. Below we state Round-Robin as a mechanism (Mechanism \ref{alg:MRR}) that takes as input a bid profile rather than the valuation functions of the agents.
In its full generality, Round-Robin should also take a permutation $N$ as an input to determine the priority of the agents. Here, for the sake of presentation, we assume that the agents in each \emph{round} (lines \ref{line:rr3}--\ref{line:rr6}) are always considered according to their ``name'', i.e., agent 1 is considered first, agent 2 second, and so on. 
This is without loss of generality, as it only requires renaming the agents accordingly. As we have mentioned in the Introduction, as an algorithm, Round-Robin outputs \efo allocations when all agents have additive valuation functions \cite{Markakis17-survey,CaragiannisKMPS19}.

	\begin{algorithm}[ht]
		\caption{Round-Robin$(\bm{b}_1, \ldots, \bm{b}_n)$ \hfill\small{ $\triangleright$ For $i\in N$, $\bm{b}_i = (b_{i1}, \ldots, b_{im})$ is the bid of agent $i$.}}
		\begin{algorithmic}[1]
			\State $S=M$\textbf{;} $(A_1,\dots,A_n) = (\emptyset,\ldots,\emptyset)$\textbf{;} $k = \lceil m/n\rceil$ 
			\For{$r = 1, \dots, k$} \Comment{{\small Each value of $r$ determines the corresponding \emph{round}. }}
			\For{$i = 1, \dots, n$} \label{line:rr3}
			\State $g = \arg\max_{h \in S}\bm{b}_i(h)$ 
			\Comment{{\small Break ties lexicographically (hence we use {\footnotesize ``$=$''} instead of {\footnotesize ``$\in$''}).}}
			\State $A_i = A_{i} \cup \{g\}$ \Comment{{\small Current agent receives (what appears to be) her favorite available good.}}
			\State $S = S\setminus \{g\}$ \Comment{{\small The good is no longer available.}}\label{line:rr6} \vspace{-2pt}
			\EndFor
			\EndFor
			\State \textbf{Return}: $\mathcal{A}=(A_1,\dots,A_n)$
		\end{algorithmic}
		\label{alg:MRR}
	\end{algorithm}

\begin{lemma}[Follows from the proof of Theorem 12.2 of \cite{Markakis17-survey}]\label{lem:ef1_of_RR}
Let $i\in N$. If $\bm{b}_i$ is the \emph{truthful} bid of agent $i$, then the allocation $\mathcal{A}$ returned by Round-Robin$(\bm{b}_1, \ldots, \bm{b}_n)$ is \efo from $i$'s perspective, i.e., for all $j\in N$, with $A_j\neq\emptyset$, there exists $g\in A_j$, such that $v_i(A_i) \geq  v_i(A_j\setminus \{g\})$. Moreover, if $i=1$, then $\mathcal{A}$ is \ef from her perspective, i.e., for all $j\in N$, $v_1(A_1) \geq  v_1(A_j)$.
\end{lemma}

Although it is long known that truth-telling is generally not a PNE in sequential allocation mechanisms (a special case of which is Round-Robin) \cite{KohlerC71}, we present here a minimal example that illustrates the mechanics of manipulation.
Let $N=\{1, 2\}$ and $M=\{a,b,c\}$ with the valuation functions being as shown in the table on the left.
The circles show the allocation returned by Round-Robin when the agents bid their true values, whereas the superscripts indicate in which order were the goods assigned. Given that agent 2 is not particularly interested in good $a$, agent 1 can manipulate the mechanism into giving her $\{a,b\}$ instead $\{a,c\}$ by claiming that these are her top goods as in the table on the right. 
\[\begin{array}{c c c c}
			& a & b & c \\
			v_1: &\circled{6}^1 & 5 & \circled{4}^3 \\   
			v_2: & 4 & \circled{6}^2 & 5 
\end{array}
\qquad\qquad\qquad
\begin{array}{c c c c}
			& a & b & c \\
			\bm{b}_1: &\circled{5}^3 & \circled{6}^1 & 4 \\   
			v_2: & 4 & 6 & \circled{5}^2 
\end{array}
\]
Thus, bidding according to $v_1, v_2$ is not a PNE. The example is minimal, in the sense that with just $1$ agent or less than $3$ goods truth-telling is a PNE of Round-Robin almost trivially.

Before moving to the main technical part of this section, 
we discuss some  assumptions that again are without loss of generality, and give an easy proof for the case of two agents.
Round-Robin as a mechanism is known to have PNE for any instance where \textit{no agent values two goods exactly the same}, and at least some such equilibria (namely, the ones consistent with the so-called \emph{bluff profile}) are easy to compute \cite{GW17}.
From a technical point of view, this assumption that all the valuation functions induce strict preference rankings is convenient, as it greatly reduces the number of corner cases one has to deal with. 
However, as we show in Theorem \ref{thm:general_bluff} in  
the Appendix,
the result of \citet{GW17} on the existence of Round-Robin's PNE extends to general additive valuation functions.
On a different but related note, we assume, for the remainder of this section, that all the bid vectors induce strict preference rankings (but not necessarily consistent with  the preference rankings induced by the corresponding valuation functions). 
This is without loss of generality, because even if a bid vector contains some bids that are equal to each other, a strict preference ranking is imposed by the lexicographic tie-breaking of the mechanism itself. So, formally, when we abuse the notation and write $g\succ_i h$ we  mean that either  $\bm{b}_i(g) > \bm{b}_i(h)$, or $\bm{b}_i(g) = \bm{b}_i(h)$ and $g$ has a lower index than $h$ in the standard naming of goods as $g_1, g_2, \ldots, g_m$.

Next, we show that for only two agents all PNE of Round-Robin are \efo with respect to the real valuation functions. To appreciate this easy result, one should compare it to the involved general proof of Theorem \ref{thm:main_theorem_for_n} in the next section, the full complexity of which seems to be necessary even for $n=3$. The straightforward but crucial observation that makes things work here is that envy-freeness and proportionality are equivalent when there are only two agents.

\begin{theorem}\label{thm:main_theorem_for_2}
	For any fair division instance $\mathcal{I}=(\{1,2\},M,\mathbf{v})$, every PNE of the Round-Robin mechanism is \efo with respect to the valuation functions $v_1,v_2$. 
\end{theorem}
	
	\begin{proof}
		Suppose towards a contradiction  that this is not the case. That is, there exists a PNE $\mathbf{b} = (\bm{b}_{1}, \bm{b}_{2})$ such that in the allocation $(A_1, A_2)$ returned by Round-Robin$(\mathbf{b})$ at least one of the agents envies the other, even after removing the most valuable good from her bundle. We will examine each agent separately. 
		
		If agent 1 does not see the allocation as \efo, then this means that she does not see it as \ef either. 
		Since envy-freeness and proportionality are equivalent for $n=2$, we get that $v_1(A_1) < v_1(M)/2$.
		According to Lemma \ref{lem:ef1_of_RR}, no matter what agent 2 bids, if agent 1 reports her true values to Round-Robin, the resulting allocation is \ef from her perspective.
		So, if $(A'_1, A'_2)$ is the allocation after agent 1 deviates to her true values, it is \ef from the point of view of the agent 1, which in turn implies that 
		$v_1(A'_1) \ge v_1(M)/2> v_1(A_1)$. This contradicts the fact that $\mathbf{b}$ is a PNE.
		
		If agent 2 does not see the allocation as \efo, then let $h_1$ be the good that agent 1 takes during the first round of round-robin, and $g^*\in  \arg\max_{h \in A_1}v_2(h)$  be the highest valued good in $A_1$ according to agent 2. Since agent 2 does not consider $(A_1, A_2)$ to be \efo, we have that $v_2(A_2)<v_2(A_1\setminus\{g^*\}) \leq v_2(A_1\setminus\{h_1\})$.
		This implies that the partition $(A_1\setminus \{h_1\}, A_2)$ of $M\setminus \{h_1\}$ is not an \ef allocation with respect to agent 2. Now we may use  a similar argument as in the previous case. First, since envy-freeness and proportionality are equivalent when $n=2$, we get that $v_2(A_2)< {v_2(M\setminus \{h_1\})}/{2}$. Then suppose agent 2 deviates to reporting her true values and let $(A'_1, A'_2)$ be the resulting allocation. Notice that the allocation of good $h_1$ is not affected by the deviation; it is still given to agent 1 during the first step of Round-Robin. From that point forward, the execution of the mechanism would be exactly the same as it would be if the input was the restrictions of $\bm{b}_{1}, v_2$ on $M\setminus \{h_1\}$ and agent 2 had higher priority than agent 1. The latter would result in an \ef allocation  with respect to agent 2 and, in particular, to the allocation $(A'_1\setminus \{h_1\}, A'_2)$. That is, we have $v_2(A'_2)\ge {v_2(A'_1\setminus \{h_1\})}$ and, therefore, $v_2(A_2)\ge {v_2(M\setminus \{h_1\})}/{2} > v_2(A_2)$. Like before, this contradicts the fact that $\mathbf{b}$ is a PNE. 
	\end{proof}

Moving to the case of general $n \ge 3$, the above simple argument no longer works. When an agent $i$ does not consider an allocation \efo because of an agent $i'$, this does not imply that $i$ got value less than ${1}/{n}$ of her value for the reduced bundle $M\setminus \{g^*\}$, where $g^*$ is her best good in $A_{i'}$. The reason for this is that  $\prop\not\Rightarrow \ef$  
anymore.

\subsection{Nash Equilibria of Round-Robin for Any Number of Agents}\label{subsec:RRn}

Here we state and prove the main result of our work.  
Despite its proof being rather involved, the intuition behind it is simple. 
As is often the case with proofs about \efo in variants of Round-Robin, the analysis boils down to arguing about agent 1 having no envy towards any other bundle.
On one hand, we know that whenever agent 1 bids truthfully, she sees the resulting allocation as being \ef (Lemma \ref{lem:ef1_of_RR}). On the other hand, no matter what agent 1 bids, we show it is possible to ``replace'' her with an imaginary version of herself who {\em (i)} does not affect the allocation, {\em (ii)} bids truthfully, and {\em (iii)} she considers the bundles of the allocation to be as valuable as the original agent 1 thought they were.
The rather elaborate formal argument relies on the recursive construction of auxiliary valuation functions and bids, done in Lemma \ref{lem:main_lemma}, and on the fact that small changes in a single preference ranking minimally change the ``history'' of available goods during the execution of the mechanism as shown in Lemma \ref{lem:swap}. For a high level description of the two lemmata, see the corresponding discussions before their statements, as well as Figure \ref{fig:main_lemma} which visualizes the main steps of the recursive construction of the alternative version of agent 1.

\begin{theorem}\label{thm:main_theorem_for_n}
	For any fair division instance $\mathcal{I}=(N,M,\mathbf{v})$, every PNE 
	of the Round-Robin mechanism 
	is \efo with respect to the valuation functions $v_1,\ldots,v_n$. 
\end{theorem}

As we will see shortly, proving Theorem \ref{thm:main_theorem_for_n} reduces to showing that the agent who ``picks first'' in the Round-Robin mechanism views the final allocation as envy-free, as long as she bids a best response to other agents' bids. 
Although Theorem \ref{thm:best_response} sounds very much like the standard statement about the value of the first agent in the algorithmic setting, its proof relies on a technical lemma that carefully builds a ``nice'' instance which is equivalent, in some sense, to the original. Recall that we have assumed that the agents' priority is indicated by their indices.

\begin{theorem}\label{thm:best_response}
	For any fair division instance $\mathcal{I}=(N,M,\mathbf{v})$, if the reported bid vector $\bm{b}_1$ of agent 1 is a best response to the (fixed) bid vectors $\bm{b}_2, \ldots, \bm{b}_n$ of all other players, then agent 1 does not envy (with respect to $v_1$) any bundle in the  allocation outputted by Round-Robin$(\bm{b}_1, \ldots, \bm{b}_n)$.
\end{theorem}

Note that since we are interested in PNE, it is always the case that each agent's bid is a best response to other agents' bids. As mentioned above,  Theorem \ref{thm:best_response} is essentially a corollary to Lemma \ref{lem:main_lemma}. The lemma  shows the existence of an alternative version of agent 1 who is truthful, her presence does not affect the original allocation, and, as long as the allocation is the same, she shares the same values with the original agent 1. Although its proof is rather involved, the high level idea is that we recursively construct a sequence of bids and valuation functions, each pair of which preserves the original allocation and the view of agent 1 for it, while being closer to being truthful. To achieve this we occasionally move value between the goods originally allocated to agent 1 and update the bid accordingly.

\begin{lemma}\label{lem:main_lemma}
	Suppose that the valuation function $v_1$ induces a strict preference ranking on the 
	goods.
	Let $\,\mathbf{b} = (\bm{b}_1, \bm{b}_2, \ldots, \bm{b}_n)$ be such that $\,\bm{b}_1$ is a best response of agent 1 to $\mathbf{b}_{-i} = (\bm{b}_2, \ldots, \bm{b}_n)$. Then there exists 
	a valuation function $v_1^*$ with the following properties:
	\begin{itemize}[labelindent=20pt,leftmargin=*,itemsep=3pt,topsep=5pt]
		\item If $\,\bm{b}_1^*= (v_{1}^*(g_1), v_{1}^*(g_2), \ldots, v_{1}^*(g_m))$, i.e., $\bm{b}_1^*$ is the truthful bid for $v_1^*$, then Round-Robin$(\mathbf{b})$ and Round-Robin$(\bm{b}_1^*, \mathbf{b}_{-1})$ produce the same allocation $(A_1,\dots,A_n)$.
		\item $v_1^*(A_1)=v_1(A_1)$. 
		\item For every good $g\in M\setminus A_1$, it holds that $v_1^*(g)=v_1(g)$.
	\end{itemize}
\end{lemma}

For the sake of presentation, we defer the proof of the lemma to the end of this section 
(as it needs an additional technical lemma that is itself quite long) 
and move to the proofs of Theorems \ref{thm:main_theorem_for_n} and \ref{thm:best_response}. In fact, given Lemma \ref{lem:main_lemma}, the two theorems are not hard to prove.
\medskip

\begin{proof}[Proof of Theorem \ref{thm:best_response}.]
	Consider an arbitrary instance $\mathcal{I}=(N,M,\mathbf{v})$ and assume that the input of Round-Robin is $\,\mathbf{b} = (\bm{b}_1, \bm{b}_2, \ldots, \bm{b}_n)$, where $\bm{b}_1$ is a best response of agent 1 to $\mathbf{b}_{-i} = (\bm{b}_2, \ldots, \bm{b}_n)$ according to her valuation function $v_1$. Let $(A_1, \dots,A_n)$ be the output of Round-Robin$(\mathbf{b})$. 
	In order to apply Lemma \ref{lem:main_lemma}, we need $v_1$ to induce a strict preference ranking over the 
	goods. 
	For the sake of presentation, we assume here that this is indeed the case, and we treat the general case formally in 
    the Appendix, as it needs an additional technical lemma (Lemma \ref{lem:ties}).
	So, we now consider the hypothetical scenario implied by Lemma \ref{lem:main_lemma} in this case: keeping agents 2 through $n$ fixed, suppose that the valuation function of agent 1 is the function $v_1^*$ given by the lemma, and her bid $\bm{b}_1^*$ is the truthful bid for $v_1^*$. The first part of Lemma \ref{lem:main_lemma} guarantees that the output of Round-Robin$(\bm{b}_1^*, \mathbf{b}_{-i})$ remains $(A_1,\dots,A_n)$.

	According to Lemma \ref{lem:ef1_of_RR}, no matter what others bid, if agent 1 (the agent with the highest priority here) reports her true values (i.e., according to $v^*_1$) to Round-Robin, the resulting allocation is \ef from her perspective.
	In our hypothetical scenario this translates into having $v_1^*(A_1)\ge v_1^*(A_i)$ for all $i \in N$. Then the second and third parts of Lemma \ref{lem:main_lemma} imply that $v_1(A_1)\ge v_1(A_i)$ for all $i \in N$, i.e., agent 1 does not envy any bundle in the original instance. 
\end{proof}

Having shown Theorem \ref{thm:best_response}, the proof of Theorem \ref{thm:main_theorem_for_n} is of similar flavour to the  proof on Round-Robin producing \efo allocations in the non-strategic setting \cite{Markakis17-survey}.\medskip

\begin{proof}[Proof of Theorem \ref{thm:main_theorem_for_n}.]
Let $\mathbf{b} = (\bm{b}_1, \bm{b}_2, \ldots, \bm{b}_n)$ be a PNE of the Round-Robin mechanism for the instance $\mathcal{I}$. By Theorem \ref{thm:best_response}, it is clear that the allocation returned by Round-Robin$(\mathbf{b})$ is \ef, and hence \efo, from the point of view of agent 1. We fix an agent $\ell$, where $\ell \ge 2$. 
For $i\in [\ell -1]$, let $h_i$ be the good that agent $i$ \textit{claims to be} her favourite among the goods that are available when it is her turn in the first round, i.e.,
$h_i= \argmax_{h\in M\setminus\{h_1\ldots,h_{i-1}\}}{b_{i}(h)}$. 
Right before agent $\ell$ is first assigned a good, all goods in $H = \{h_1, \ldots, h_{\ell - 1}\}$ have already been allocated. We are going to consider the instance $\mathcal{I'}=(N',M',\mathbf{v}')$ in which all goods in $H$ are missing. That is, $N'=N$, $M' = M\setminus H$, and $\mathbf{v}' = (v'_1, \ldots, v'_n)$ where $v'_i = v_i|_{M'}$, for $i\in [n]$, is the restriction of the function $v_i$ on $M'$. Similarly define $\bm{b}'_i = \bm{b}_i|_{M'}$, for $i\in [n]$, the restrictions of the bids to the available goods, and $\,\mathbf{b}' = (\bm{b}'_1, \ldots, \bm{b}'_n)$. Finally, we consider the version of Round-Robin, call it Round-Robin$_{\ell}$, that starts with agent $\ell$ and then follows the indices in increasing order.

We claim that for Round-Robin$_{\ell}$ the bid $\bm{b}'_{\ell}$ is a best response for agent $\ell$ assuming that the restricted bid vectors of all the other agents are fixed. To see this, notice that for any $\bm{c}_{\ell} = (c_{\ell 1}, c_{\ell 2}, \ldots, c_{\ell m})$, the bundles given to agent $\ell$ by Round-Robin$(\bm{c}_{\ell}, \mathbf{b}_{-\ell})$ and Round-Robin$_{\ell}(\bm{c}_{\ell}|_{M'}, \mathbf{b}'_{-\ell})$ are the same! In fact, the execution of Round-Robin$_{\ell}(\bm{c}_{\ell}|_{M'}, \mathbf{b}'_{-\ell})$ is identical to the execution of Round-Robin$(\bm{c}_{\ell}, \mathbf{b}_{-\ell})$ from its $\ell$th step onward. So, if  $\bm{b}'_{\ell}$ was not a best response in the restricted instance, then there would be a profitable deviation for agent $\ell$, say $\bm{b}^*_{\ell}$, so that $\ell$ would prefer her bundle in Round-Robin$_{\ell}(\bm{b}^*_{\ell}, \mathbf{b}'_{-\ell})$ to her bundle in Round-Robin$_{\ell}(\mathbf{b}')$. This would imply that any extension of $\bm{b}^*_{\ell}$ to a bid vector for all goods in $M$ (by arbitrarily assigning numbers to goods in $H$) would be a profitable deviation for agent $\ell$ in the profile $\mathbf{b}$ for Round-Robin, contradicting the fact that $\mathbf{b}$ is a PNE.

Now we may apply Theorem \ref{thm:best_response} for Round-Robin$_{\ell}$ (where agent $\ell$ plays the role of agent 1 of the theorem's statement) for instance $\mathcal{I'}$ and bid profile $\mathbf{b}'$. The theorem implies that  agent $\ell$ does not envy any bundle in the allocation $(A_1, \dots,A_n)$ outputted by Round-Robin$_{\ell}(\mathbf{b}')$, i.e., $v'_{\ell}(A_\ell) \ge v'_{\ell}(A_i)$, for all $i\in [n]$. Using the observation made above about the execution of Round-Robin$_{\ell}(\mathbf{b}')$ being identical to the execution of Round-Robin$(\mathbf{b})$ after $\ell-1$ goods have been allocated, we have that Round-Robin$(\mathbf{b})$ returns the allocation $(A_1\cup\{h_1\}, \dots,A_{\ell -1}\cup\{h_{\ell -1}\}, A_{\ell}, \ldots, A_n)$. So, for any $i<\ell$ we have
$v_{\ell}(A_\ell) = v'_{\ell}(A_\ell) \ge v'_{\ell}(A_i)  = v_{\ell}(A_i) 
= v_{\ell}((A_i \cup \{h_i\}) \setminus \{h_i\})$,
whereas for $i>\ell$ we simply have
$v_{\ell}(A_\ell) = v'_{\ell}(A_\ell) \ge v'_{\ell}(A_i) = v_{\ell}(A_i)$.
Thus, the allocation returned by Round-Robin$(\mathbf{b})$ is \efo from the point of view of agent $\ell$. 
\end{proof}

Before we move on to the proof of Lemma \ref{lem:main_lemma}, we state another technical lemma. 
Suppose an agent changes her bid so that in her preference ranking a single good is moved down the ranking, and then---keeping everything else fixed---we run Round-Robin on the new instance. Surprisingly, Lemma \ref{lem:swap} states that, in any step, the set of available goods 
differs by at most one good from the corresponding set in the original run of Round-Robin. To formalize this, we need some additional notation and terminology.

\smallskip\begin{definition}
	Let $\succ$ and $\succ'$ be two strict preference rankings on $M$ and $\{q_1, q_2, \ldots, q_m\}$ be a renaming of the goods according to $\succ$, i.e., $q_1 \succ q_2 \succ \ldots \succ q_m$. We say that $\succ$ and $\succ'$ are \emph{within a partial slide of each other} 
	if there exist $x,y \in [m]$, $x<y$, such that 
\[q_1 \succ' \ldots \succ' q_{x-1} \succ' q_{x+1} \succ'  \ldots \succ' q_{y} \succ' q_x \succ' q_{y+1} \succ' \ldots  \succ'q_m\,.\]
\end{definition}

Also, given a profile $\mathbf{b} = (\bm{b}_1, \ldots, \bm{b}_n)$, let $M_t(\mathbf{b})$ denote the set of available goods right after $t-1$ goods have been allocated in a run of Round-Robin$(\mathbf{b})$.

When we run Round-Robin on two profiles which induce the same preference rankings for all agents but one, and for this agent the two preference rankings are within a partial slide of each other, then the resulting allocations may different drastically. Yet, as the next lemma states, the available goods at every step are almost the same in the two executions of the mechanism. What happens, roughly speaking, is that at the beginning of each step there is at most one difference between the sets of unallocated goods, and this is difference is either ``fixed'' or ``passed on'' to the next step, possibly slightly altered.

\begin{lemma}\label{lem:swap}
	Let $\mathbf{b} = (\bm{b}_1, \ldots, \bm{b}_n)$ and $\mathbf{b}' = (\bm{b}'_i,\mathbf{b}_{-i})$ be two profiles such that the corresponding induced preference rankings $\succ_i$ and $\succ'_i$ of agent $i$ are within a partial slide of each other. Then $|M_t(\mathbf{b}) \setminus M_t(\mathbf{b}')| = |M_t(\mathbf{b}') \setminus M_t(\mathbf{b})| \le 1$ for all $t\in [m+1]$. 
\end{lemma}

\begin{proof}
Clearly, for $t\le i$ we have $M_t(\mathbf{b}) = M_t(\mathbf{b}')$ as the runs of Round-Robin$(\mathbf{b})$ and Round-Robin$(\mathbf{b}')$ are identical at least up to the allocation of the first $i-1$ goods. We are going to prove the statement by induction on $t$ using this observation as our base case. Assume that for some $t\ge i$, $|M_t(\mathbf{b}) \setminus M_t(\mathbf{b}')| = |M_t(\mathbf{b}') \setminus M_t(\mathbf{b})| \le 1$. Up to this point, $t-1$ goods have been allocated already. Let $j$ be the next agent to get a good and let $g$ (resp.~$g'$) be this good in Round-Robin$(\mathbf{b})$ (resp.~in Round-Robin$(\mathbf{b}')$).
The only challenging (sub)case is when $M_t(\mathbf{b})$ and $M_t(\mathbf{b}')$ each contain one non-common element and neither of these two elements is about to be allocated in the corresponding run of Round-Robin.
\medskip

	\noindent\ul{Case 1 ($M_t(\mathbf{b}) = M_t(\mathbf{b}')$).} 
	No matter who $j$ is and what $g$ and $g'$ are, it is straightforward to see that either $M_{t+1}(\mathbf{b}) \setminus M_{t+1}(\mathbf{b}') = M_{t+1}(\mathbf{b}') \setminus M_{t+1}(\mathbf{b}) = \emptyset$ (when $g=g'$), or $M_{t+1}(\mathbf{b}) \setminus M_{t+1}(\mathbf{b}') = \{g'\}$ and $M_{t+1}(\mathbf{b}') \setminus M_{t+1}(\mathbf{b}) = \{g\}$ (when $g\neq g'$). Thus, $|M_{t+1}(\mathbf{b}) \setminus M_{t+1}(\mathbf{b}')| = |M_{t+1}(\mathbf{b}') \setminus M_{t+1}(\mathbf{b})| \le 1$. \smallskip

	Before we move to Case 2, it is important to take a better look on how can we move away from Case 1 for the very first time. 
    That is, we want to focus on the first time step when the good allocated in Round-Robin$(\mathbf{b})$ is different from the good allocated in Round-Robin$(\mathbf{b}')$, if such a time step exists for the specific profiles.
    Since $\succ_i$ and $\succ'_i$ are within a partial slide of each other, there exists a unique
    good $s\in M$ that goes from a better position in $\succ_i$ to a worse position in $\succ'_i$. The next claim about $s$ is crucial for showing that the last subcase of Case 2 below cannot happen.

    \begin{claim}\label{claim:first_difference}
    Suppose that $t_*$ is the first time step where the good $\gamma$ allocated in Round-Robin$(\mathbf{b})$ is different from the good $\gamma'$ allocated in Round-Robin$(\mathbf{b}')$. Then, $\gamma=s$.
    \end{claim}

    \begin{proof}[Proof of Claim \ref{claim:first_difference}]
    We begin with the observation that $t_*$ cannot be first time step when $\gamma\neq \gamma'$ if $j\neq i$ at this point. Indeed, if it was  $j\neq i$, since $M_{\ell}(\mathbf{b}) = M_{\ell}(\mathbf{b}')$ for all $\ell\in [t_*]$ and the induced preference ranking of  $j$ in this case is the same in both $\mathbf{b}$ and $\mathbf{b}'$, we have that the two runs of Round-Robin should make the same choice for $j$ in the time step $t_*$; that would contradict the choice of $t_*$ itself. So, after the same $t_*-1$ goods have been allocated by Round-Robin$(\mathbf{b})$ and Round-Robin$(\mathbf{b}')$, agent $i$ is about to be given $\gamma$ and $\gamma'$ respectively in the two runs from the set $M_{t_*} = M_{t_*}(\mathbf{b})  = M_{t_*}(\mathbf{b}')$ of available goods. We are going to show that these goods cannot be arbitrary. 
	Recall that $\gamma$ is the best good in $M_{t_*}$ with respect to $\succ_i$; similarly for $\gamma'$ and  $\succ'_i$.
	First, notice that $\succ_i$ and $\succ'_i$ are identical on $M_{t_*} \setminus \{s\}$ and, thus, for  $\gamma$ and $\gamma'$ to be distinct at least one of them must be $s$. 
	Since $\gamma\neq \gamma'$, either $\gamma=s$ or $\gamma'=s$ but not both. Assume for a contradiction that $\gamma'=s$ and $\gamma=x\neq s$. Since $s\in M_{t_*}$, it is available to both. 
	The fact that $x\neq s$ implies that $x\succ_i s$. However, this also mean $x\succ'_i s$ which, given the availability of $x$, contradicts the choice of $\gamma'$. 
	We conclude that $\gamma=s$ and $\gamma'\neq s$. \renewcommand\qedsymbol{{\footnotesize $\boxdot$}}
    \end{proof}
	\medskip

    \noindent\ul{Case 2 ($M_t(\mathbf{b}) \setminus M_t(\mathbf{b}') = \{h\}$ and $M_t(\mathbf{b}') \setminus M_t(\mathbf{b}) = \{h'\}$).} 
	When $g=h$ or $g'=h'$, it is very easy to complete the inductive step. First, if $g=h$ \emph{and} $g'=h'$, then we immediately get $M_{t+1}(\mathbf{b}) = M_{t+1}(\mathbf{b}')$. Further, if $g=h$ and $g'\neq h'$, then we have that 
	\[M_{t+1}(\mathbf{b}) \setminus M_{t+1}(\mathbf{b}') = (M_t(\mathbf{b}) \setminus \{h\}) \setminus (M_t(\mathbf{b}') \setminus \{g'\})= ((M_t(\mathbf{b}) \setminus M_t(\mathbf{b}'))  \setminus \{h\})\cup \{g'\} =  \{g'\}\,,\]
	where the second equality holds because $g'\in M_t(\mathbf{b}) \cap M_t(\mathbf{b}')$ in this case, and
	\[M_{t+1}(\mathbf{b}') \setminus M_{t+1}(\mathbf{b}) = (M_t(\mathbf{b}') \setminus \{g'\}) \setminus (M_t(\mathbf{b}) \setminus \{h\})= M_t(\mathbf{b}') \setminus M_t(\mathbf{b})  =  \{h'\}\,,\]
	where here the second equality holds because $g'\in M_t(\mathbf{b})$ and $h\notin M_t(\mathbf{b}')$. The subcase where $g\neq h$ and $g'= h'$ is symmetric and we similarly get
	\[M_{t+1}(\mathbf{b}) \setminus M_{t+1}(\mathbf{b}')  =  \{h\} \text{\ \ \ \ and\ \ \ \ } M_{t+1}(\mathbf{b}') \setminus M_{t+1}(\mathbf{b}) =  \{g\}\,.\]

	It remains to deal with the subcase where $g\neq h$ and $g'\neq h'$. If $g=g'$, then we immediately get $M_{t+1}(\mathbf{b}) \setminus M_{t+1}(\mathbf{b}') = M_t(\mathbf{b}) \setminus M_t(\mathbf{b}') = \{h\}$ and $M_{t+1}(\mathbf{b}') \setminus M_{t+1}(\mathbf{b}) = M_t(\mathbf{b}') \setminus M_t(\mathbf{b}) = \{h'\}$. So, we may assume that $h\neq g \neq g'\neq h'$. We are going to show that this cannot actually happen, as it would lead to a contradiction. Notice that $h\neq g \neq g'\neq h'$ implies $g, g' \in M_t(\mathbf{b}) \cap M_t(\mathbf{b}')$. If agent $j$ is different than agent $i$, this would mean that $g \succ_j g'$ and $g' \succ_j g$ because of the corresponding choices of the algorithm when the input is $\mathbf{b}$ and $\mathbf{b}'$ respectively (recall that the bid, and thus the induced preference ranking, of $j$ is the same in both profiles); that would be a contradiction. Therefore, it must be the case that $j = i$. 
	Since we are in Case 2, a scenario leading to Case 2 for the first time (as described in Claim \ref{claim:first_difference}) must have already happened. Consequently, by Claim \ref{claim:first_difference}, $s$ is not available at this point in $M_t(\mathbf{b})$ and hence $s\notin M_t(\mathbf{b}) \cap M_t(\mathbf{b}')$. This means that $\{g, g'\} \subseteq M\setminus\{s\}$ and, therefore, $g$ and $g'$ have the same ordering in both preference rankings of agent $i$. That is, $g \succ_i g'$ implies $g \succ'_i g'$, contradicting the optimality of $g'$ in $M_t(\mathbf{b}')$ with respect to $\succ'_i$.
	\medskip

	We conclude that in any possible case, $|M_{t+1}(\mathbf{b}) \setminus M_{t+1}(\mathbf{b}')| = |M_{t+1}(\mathbf{b}') \setminus M_{t+1}(\mathbf{b})| \le 1$. This concludes the induction. 
\end{proof}	

We are now ready to prove Lemma \ref{lem:main_lemma}. As it was noted before the lemma's statement, we will occasionally move value among the goods allocated to agent 1. This is when Lemma \ref{lem:swap} is crucial. It allows us to guarantee that there is sufficient value for satisfying all the desired properties of the intermediate valuation functions we define. \smallskip

\begin{proof}[Proof of Lemma \ref{lem:main_lemma}.]
Recall that $k = \lceil m/n\rceil$, i.e., we have $k$ total rounds. 
Let $\succ_1$ be the preference ranking induced by $\bm{b}_1$ and consider all the goods according to this ranking: $h_1 \succ_1 h_2 \succ_1 \ldots \succ_1 h_m$. Let  $n_1 = 1 <n_2 <\dots< n_k$ be the indices in this ordering of the goods assigned to agent 1 by Round-Robin$(\mathbf{b})$, i.e., in round $r$ agent 1 receives good $h_{n_r}$. This means that $A_1=\{h_{n_1}, \ldots, h_{n_{k}}\}$.

We will recursively construct $v_1^*$ from $v_1$, over the rounds of Round-Robin. In particular, we are going to define a sequence of intermediate bid vectors $\bm{b}^r_1$ and valuation functions $v^r_1$, one for each round $r$ starting from the last round $k$, so that $v_1^* = v_1^1$ and $\bm{b}^*_1 = \bm{b}^1_1$. For defining each $\bm{b}^r_1$ we typically use a number of auxiliary bid vectors to break down and better present the construction.
Also, for any round $r$, we are going to maintain that

\begin{enumerate}[labelindent=20pt,leftmargin=*,itemsep=4pt,topsep=5pt]
\item[(i)] $v_1^r(A_1)=v_1(A_1)$.
\item[(ii)] $v_1^r(g)=v_1(g)$, for any $g\in M\setminus A_1$.
\item[(iii)] $\bm{b}^r_1$ is \emph{truthful from round $r$} with respect to $v^r_1$, meaning that for every good that is no better than $h_{n_r}$, according to the preference ranking $\succ^r_1$ induced by $\bm{b}^r_1$, we have that its bid matches its value; formally, $g\nsucc^r_1 h_{n_r} \Rightarrow \bm{b}^r_1(g) = v^r_1(g)$. 
\item[(iv)] The preference ranking ${\succ}^r_1$ (induced by ${\bm{b}}^r_1$) is identical to ${\succ}_1$ (induced by ${\bm{b}}_1$) up to good $h_{n_{r-1}}$. 
\item[(v)] $\min_{g, h\in M,\, g\neq h}|v^{r}_1(g)-v^{r}_1(h)|>0$.
\end{enumerate}
Carefully ensuring that all five properties hold, makes the formal construction rather complicated. We provide a visual abstraction of the high level idea of a single step in the recursive construction of $v_1^*$ in Figure \ref{fig:main_lemma}; see the caption for a detailed connection with the formal steps of the proof.

Let us focus on round $k$, i.e., the last round. Let $\lambda_k$ be the most valuable (according to $v_1$) available good at the very beginning of the round. It is easy to see that $h_{n_k} = \lambda_k$; if not, then by increasing her bid for $\lambda_k$ to be slightly above her bid for $h_{n_k}$ agent 1 would end up with the bundle $\{h_{n_1}, \ldots, h_{n_{k-1}}, \lambda_k\}$ which is a strict improvement over $A_1$ and would contradict the fact that $\bm{b}_1$ is a best response of agent 1.
We construct the auxiliary bid $\bar{\bm{b}}^k_1$ by ``moving up'' in $\succ_1$ every good that is more valuable than $\lambda_k$ but comes after it in $\succ_1$ (i.e., the purple blocks in Figure \ref{fig:main_lemma}). Formally, $(\lambda_k \succ_1 g) \,\wedge\, (v_1(g)> v_1(\lambda_k)) \Rightarrow \bm{b}_1(h_{n_k}) < \bar{\bm{b}}^k_1(g) < \bm{b}_1(h_{n_k -1})$, where these bids are chosen arbitrarily, as long as they are distinct from each other.
Note that this small modification does not affect the allocation at all. Indeed, every good the bid of which was improved is still worse than $h_{n_k -1}$ in the preference ranking $\bar{\succ}^k_1$ induced by $\bar{\bm{b}}^k_1$, so no decision in rounds $1,\ldots,k-1$ is affected and, by the definition of $\lambda_k$, these goods were not actually available for agent 1 in the beginning of round $k$, so the decisions in round $k$ are not affected either.
Next we define ${\bm{b}}^k_1$ by replacing the bids with the actual values for every good that is no better than $h_{n_k}$ in $\bar{\succ}^k_1$,
as well as by scaling the bids of all other goods to remain larger than $\bm{b}^k_1(h_{n_k})$, if necessary.
Although the latter can be done in several ways, we can simply multiply bids by $\bm{b}^k_1(h_{n_k})/ \bar{\bm{b}}^k_1(h_{n_k})$. Formally,  ${\bm{b}}^k_1$ is  defined by
\[ g  \not\!\!{\bar{\succ}}^k_1 \, h_{n_k} \Rightarrow \bm{b}^k_1(g) = v_1(g) 
\text{\ \ \ \ and\ \ \ \ } g \,\,\bar{\succ}^k_1 \,h_{n_k} \Rightarrow \bm{b}^k_1(g) = \bar{\bm{b}}^k_1(g)\cdot \bm{b}^k_1(h_{n_k}) / \bar{\bm{b}}^k_1(h_{n_k})\,.\]
Note that the preference ranking ${\succ}^k_1$ induced by ${\bm{b}}^k_1$ is identical to $\bar{\succ}^k_1$ up to good $h_{n_k - 1}$, and that $\lambda_k$ is the good with the highest bid in ${\bm{b}}^k_1$ among the goods that are available in the last round.
Hence, Round-Robin$(\bm{b}_1^k, \mathbf{b}_{-1})$ still produces the allocation $(A_1,\dots,A_n)$.
Also, recall that $\bar{\succ}^k_1$ is identical to ${\succ}_1$ up to good $h_{n_k - 1}$ and, thus, up to at least good $h_{n_{k-1}}$, implying that ${\succ}^k_1$ satisfies property (iv) above.
Finally, by setting $v^k_1 = v_1$, it is clear that $\bm{b}^k_1$ is truthful from round $k$ with respect to $v^k_1$, but also that $\min_{g, h\in M, g\neq h}|v^{k}_1(g)-v^{k}_1(h)|>0$, $v_1^k(A_1)=v_1(A_1)$, and $v_1^k(g)=v_1(g)$, for all $g\in M\setminus A_1$. That is, all properties (i)-(v) are satisfied.\medskip

Moving to an arbitrary round $r< k$ we are going to follow a similar, albeit a bit more complicated, approach, where now it will be necessary to move value among the goods of $A_1$.
So, assume that $\bm{b}^{r+1}_1$ and $v^{r+1}_1$ have already been constructed and have the desired properties (i)--(v) mentioned above, and let $\succ^{r+1}_1$ be the preference ranking induced by $\bm{b}^{r+1}_1$. 
Consider the execution of Round-Robin$(\bm{b}_1^{r+1}, \mathbf{b}_{-1})$.
For $i\ge r$, let $\lambda_i$ be the most valuable available good with respect to $v^{r+1}_1$ at the very beginning of round $i$ and $\ell_i$ be the most valuable good with respect to $v^{r+1}_1$ (or equivalently with respect to $v_1$ as $\ell_i\in M\setminus A_1$) that is allocated to some \emph{other} agent during round $i$. 
By property (iii) of $\bm{b}^{r+1}_1$ and $v^{r+1}_1$ we know that in future rounds agent 1 will have $\lambda_{r+1}=h_{n_{r+1}}, \lambda_{r+2} = h_{n_{r+2}}, \ldots, \lambda_{k} = h_{n_{k}}$  allocated to her. By property (iv) of $\bm{b}^{r+1}_1$ we further know that in the current round agent 1 is going to get good $h_{n_r}$. Unlike what happened for round $k$, however, here $h_{n_r}$ may be different from $\lambda_r$. We will  
consider two cases depending on this.

First, though, similarly to what we did before, we define the auxiliary bid $\bar{\bar{\bm{b}}}^r_1$ by setting $\bm{b}^{r+1}_1(h_{n_r}) \allowbreak < \allowbreak  \bar{\bar{\bm{b}}}^r_1(g) < \allowbreak  \bm{b}^{r+1}_1(h_{n_r -1})$ for all goods $g$ such that $h_{n_r} \succ^{r+1}_1 g$ and $v^{r+1}_1(g)> v_1(\lambda_r)$;  these $\bar{\bar{\bm{b}}}^r_1$ entries are arbitrary, as long as they satisfy the inequalities and are distinct from each other.
By now it should be clear that moving from $\bm{b}^{r+1}_1$ to $\bar{\bar{\bm{b}}}^r_1$ does not affect the allocation since every good that had its bid improved is still worse than $h_{n_r -1}$ in the preference ranking $\ \bar{\bar{\!\!\!\succ}}^r_1$ induced by $\bar{\bar{\bm{b}}}^r_1$ 
and, by the definition of $\lambda_r$, these goods were already not available in the beginning of round $r$. 
That is, Round-Robin$\big(\bar{\bar{\bm{b}}}_1^r, \mathbf{b}_{-1}\big)$ returns $(A_1,\dots,A_n)$. \medskip

\noindent\ul{Case 1 ($h_{n_r} = \lambda_r$).}  
This case is similar to what we did for round $k$. We go straight from $\bar{\bar{\bm{b}}}^r_1$ to ${\bm{b}}^r_1$ by replacing the bids with the corresponding $v^{r+1}_1$ values for all goods that are no better than $h_{n_r}$  in $\ \bar{\bar{\!\!\!\succ}}^r_1$,
and by scaling the bids of all other goods to remain larger than $v^{r+1}_1(\lambda_{r})=\bm{b}^r_1(\lambda_{r})=\bm{b}^r_1(h_{n_r})$. Formally, we have
\[ g \not{\bar{\bar{\!\!\!\succ}}}^r_1 \, h_{n_r} \Rightarrow \bm{b}^r_1(g) = v^{r+1}_1(g) \text{\ \ \ \ and\ \ \ \ } g\ {\ \bar{\bar{\!\!\!\succ}}}^r_1 \,h_{n_r} \Rightarrow \bm{b}^r_1(g) = \bar{\bar{\bm{b}}}^r_1(g)\cdot v^{r+1}_1(\lambda_{r}) / \bar{\bar{\bm{b}}}^r_1(h_{n_r})\,.\]
The preference ranking ${\succ}^r_1$ induced by ${\bm{b}}^r_1$ is identical to $\ {\bar{\bar{\!\!\!\succ}}}^r_1$ up to good $h_{n_r}$, so Round-Robin$(\bm{b}_1^r, \mathbf{b}_{-1})$ up to the beginning of round $r$ still allocates $\{h_{n_1}, \ldots, h_{n_{r-1}}\}$ to agent 1 in that order. Also, from good $h_{n_r}$ onward, ${\succ}^r_1$ is defined in such a way that the best available good in the beginning of round $i\ge r$ with respect to ${\succ}^r_1$ is $h_{n_i}$. Therefore, the final bundle for agent 1 is still $A_1$ and the overall allocation is still $(A_1,\dots,A_n)$ as $\mathbf{b}_{-1}$ is fixed and goods in $A_1$ are allocated in the exact same order.
Moreover, recall that ${\ \bar{\bar{\!\!\!\succ}}}^r_1$ is identical to ${\succ}_1$ up to good $h_{n_{r-1}}$ (in fact, up to good $h_{n_r - 1}$), implying that ${\succ}^r_1$ satisfies property (iv).
Given that no changes to values were necessary and that we made the relevant (for rounds $r, \ldots, k$) entries of ${\bm{b}}^r_1$ equal to the corresponding $v^{r+1}_1$ values, we may set $v^r_1 = v^{r+1}_1$ to get that $\bm{b}^r_1$ is truthful from round $r$ with respect to $v^r_1$, but also that $\min_{g, h\in M, g\neq h}|v^{r}_1(g)-v^{r}_1(h)|>0$, $v_1^r(A_1)=v_1(A_1)$, and $v_1^r(g)=v_1(g)$, for all $g\in M\setminus A_1$. \medskip

\begin{figure}[t]
    \centering
    \includegraphics[width=\textwidth]{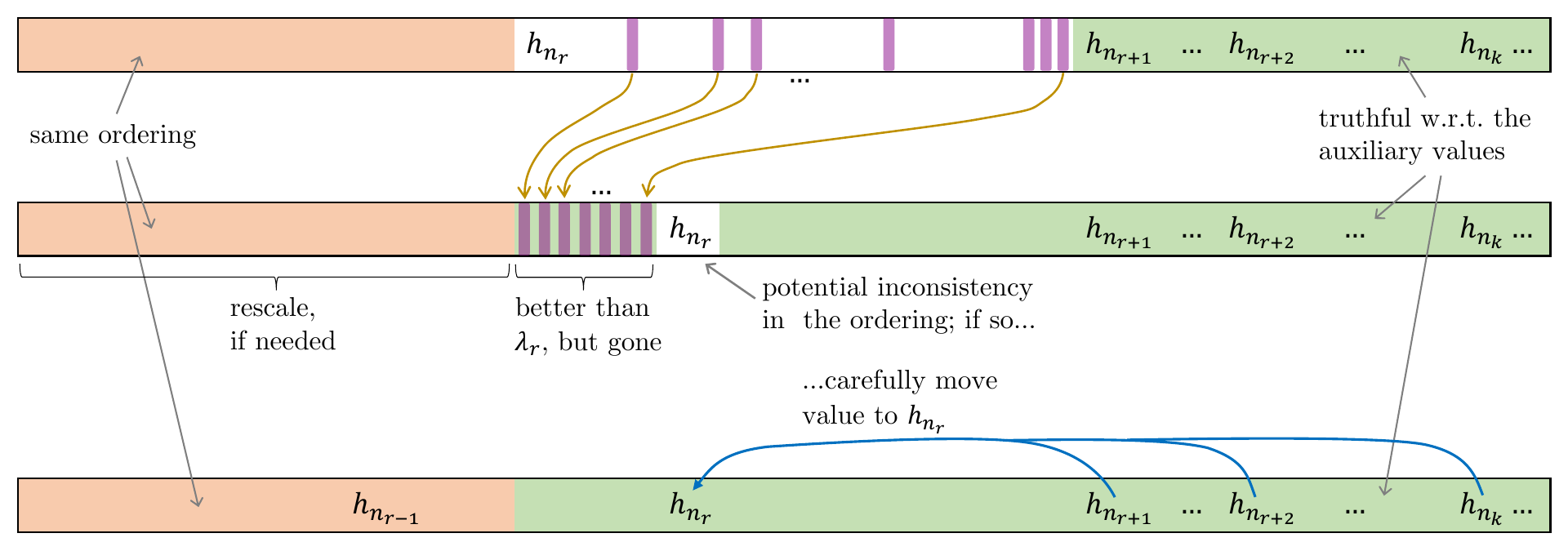}
    \caption{Each row is a visual abstraction of the preference ranking of the corresponding auxiliary bid during the transformation we perform for round $r$. Note that for $r = k$, or an arbitrary $r < k$ that falls under Case 1, only the first two rows are relevant. The purple blocks represent the goods which are more valuable than $\lambda_r$ but come after it in $\succ_1$, i.e., they are all unavailable at the beginning of round $r$. Moving those up to the left of $h_{n_r}$ corresponds to constructing $\bar{\bar{\bm{b}}}^r_1$ (or ${\bar{\bm{b}}}^r_1$ for round $k$). Then,  we replace the bids with the corresponding $v^{r+1}_1$ values for all goods on the green part and scale the bids  of the goods in the orange part to keep the ordering consistent. For $r = k$ or in Case 1 for $r < k$, we already obtain ${\bm{b}}^r_1$ this way and we are done for the round. In Case 2, however, there still are goods on the right of $h_{n_r}$ that are better (at least one, as $\lambda_r \neq h_{n_r}$ in this case), yet we want to keep $h_{n_r}$ in place. This temporary inconsistency is represented by ${\bar{\bm{b}}}^r_1$ in Case 2. In order to fix that and move to ${\bm{b}}^r_1$, we transfer enough value from $h_{n_{r+1}}, h_{n_{r+2}}, \ldots, h_{n_{k}}$ to $h_{n_r}$, without compromising their position in the ordering; Claim \ref{claim:move_value} guarantees that the latter is always possible in this case.
    }
    \label{fig:main_lemma}
\end{figure}

\noindent\ul{Case 2 ($h_{n_r} \neq \lambda_r$).}  
Here we are going to move value from goods $\lambda_{r+1}, \lambda_{r+2}, \ldots, \lambda_{k}$ to $h_{n_{r}}$ while defining $v^{r}_1$.
The main idea is that we would like $h_{n_r}$ to become the most valuable available good at the beginning of round $r$  with respect to $v^r_1$, although this is not the case for $v^{r+1}_1$ as $v^{r+1}_1(h_{n_r})<v^{r+1}_1(\lambda_r)$. The constraints we need to satisfy make this task rather tricky: properties (i) and (ii) must hold, so value can only be transferred between goods of $A_1$, but this should happen in a way that ensures that in future rounds the goods given to agent 1 remain $h_{n_{r+1}}, \ldots, h_{n_{k}}$ in that order.

We begin with a rather benign modification of $\bar{\bar{\bm{b}}}^r_1$, which is almost identical to what we did in Case 1, except that we do not update the bid of $h_{n_r}$ with its $v^{r+1}_1$ value. We do this to make sure that $h_{n_r}$ still seems like the most attractive good of round $r$ and the overall allocation remains the same. Specifically, we define the auxiliary bid ${\bar{\bm{b}}}^r_1$ by
\begin{equation}\label{eq:b-bar}
h_{n_r}\  \,{\bar{\bar{\!\!\!\succ}}}^r_1 \  g \Rightarrow \bar{\bm{b}}^r_1(g) = v^{r+1}_1(g) \text{\ \ \ \ and\ \ \ \ } h_{n_r} \not{\bar{\bar{\!\!\!\succ}}}^r_1 \ g \Rightarrow \bar{\bm{b}}^r_1(g) = \bar{\bar{\bm{b}}}^r_1(g)\cdot (v^{r+1}_1(\lambda_{r}) + \delta/2) / \bar{\bar{\bm{b}}}^r_1(h_{n_r})\,,
\end{equation}
where $\delta = \min_{g, h\in M, g\neq h}|v^{r+1}_1(g)-v^{r+1}_1(h)|>0$.
It is easy to check that Round-Robin$(\bar{\bm{b}}_1^r, \mathbf{b}_{-1})$ returns $(A_1,\dots,A_n)$.
Indeed, the preference ranking $\bar{{\succ}}^r_1$ induced by $\bar{\bm{b}}^r_1$ is identical to $\ {\bar{\bar{\!\!\!\succ}}}^r_1$ up to good $h_{n_r}$, so agent 1 receives $h_{n_1}, \ldots, h_{n_{r}}$ in the first $r$ rounds, whereas any bid that was higher than $v^{r+1}_1(h_{n_{r+1}})$ and has been updated  to its $v^{r+1}_1$ value is not available at the beginning of round $r+1$ anyway. The latter is true, because otherwise such a good would have been chosen by Round-Robin$(\bar{\bm{b}}_1^r, \mathbf{b}_{-1})$ and Round-Robin$({\bm{b}}_1^r, \mathbf{b}_{-1})$ instead of $h_{n_{r+1}}$.

Having ${\bar{\bm{b}}}^r_1$ as a point of reference, we now take a closer look to what happens if, starting at round $r$, agent 1 would receive her goods according to $v^{r+1}_1$. Note that this would be the same as just changing  ${\bar{\bm{b}}}^r_1(h_{n_r})$ to $v^{r+1}_1(h_{n_r})$. We call this new auxiliary bid $\hat{\bm{b}}^r_1$; note that this is the first time we introduce a bid that \textit{does not} preserve the original allocation.
Similarly to our definition of the $\lambda_i$s, we consider the execution of Round-Robin$(\hat{\bm{b}}^r_1, \mathbf{b}_{-1})$ and define $\hat{\lambda}_i$ to be the most valuable available good with respect to $v^{r+1}_1$ at the very beginning of round $i$, for $i\ge r$.
While we know that $\hat{\lambda}_r = \lambda_r$, in general we  have no reason to expect that $\hat{\lambda}_i$ and $\lambda_i$ are the same.
Actually, the fact that $\bm{b}_1$---and thus $\bar{\bar{\bm{b}}}^r_1$---is a best response, combined with $h_{n_r} \neq \lambda_r$, imply that
\begin{equation}\label{eq:sum_of_lambdas}
v^{r+1}_1(h_{n_r}) + \sum_{i= r+1}^{k} v^{r+1}_1(\lambda_i) > \sum_{i= r}^{k} v^{r+1}_1(\hat{\lambda}_i) \,.
\end{equation}

Coming back to the challenge of moving value from $v^{r+1}_1(\lambda_{r+1}), \ldots, v^{r+1}_1(\lambda_{k})$ to $v^{r+1}_1(h_{n_{r}})$ (and equally so from $\hat{\bm{b}}^{r+1}_1(\lambda_{r+1}), \ldots, \hat{\bm{b}}^{r+1}_1(\lambda_{k})$ to $\hat{\bm{b}}^{r+1}_1(h_{n_{r}})$), we want to make sure that enough value can be moved to eventually get $v^{r}_1(h_{n_{r}})$ slightly above $v^{r+1}_1(\lambda_{r})$ while each $h_{n_{i}}$ maintains more value than any other available good in round $i>r$. 
\begin{claim}\label{claim:move_value}
There exists $\varepsilon > 0$ such that
\[ \sum_{i= r+1}^{k} \big( v^{r+1}_1(\lambda_i) -  \max\big\{v^{r+1}_1(h_{n_{i+1}}), v^{r+1}_1(\ell_i)\big\} \big) =  v^{r+1}_1({\lambda}_r) - v^{r+1}_1(h_{n_r}) + \varepsilon \,.\]
\end{claim}

\begin{proof}[Proof of Claim.]
Let $\varepsilon = v^{r+1}_1(h_{n_r}) + \sum_{i= r+1}^{k} v^{r+1}_1(\lambda_i) - \sum_{i= r}^{k} v^{r+1}_1(\hat{\lambda}_i)$;  the fact that $\varepsilon>0$ follows from inequality \eqref{eq:sum_of_lambdas}. Next notice that the preference rankings $\bar{{\succ}}^r_1$ and $\hat{{\succ}}^r_1$, induced by ${\bar{\bm{b}}}^r_1$ and ${\hat{\bm{b}}}^r_1$ respectively, are within a partial slide from each other! Indeed,  $h_{n_{r}}$ is moved to a worst position in $\hat{{\succ}}^r_1$ compared to $\bar{{\succ}}^r_1$, but otherwise the two preference rankings are the same.
Besides the easy observation that Round-Robin$(\hat{\bm{b}}^r_1, \mathbf{b}_{-1})$ and Round-Robin$(\bar{\bm{b}}^r_1, \mathbf{b}_{-1})$ run identically for $i-1$ rounds,
this also means that Lemma \ref{lem:swap} applies.  
That is, we have that 
in each round $i$ of Round-Robin$(\hat{\bm{b}}^r_1, \mathbf{b}_{-1})$, for $i\ge r$, there is at most one good that is unavailable despite being available in round $i$ of Round-Robin$(\bar{\bm{b}}^r_1, \mathbf{b}_{-1})$. In particular, in round $i$ of Round-Robin$(\hat{\bm{b}}^r_1, \mathbf{b}_{-1})$, for $i\ge r$, at least two goods from $\{h_{n_i}, h_{n_{i+1}}, \ell_i\}$ are available, where conventionally we define $h_{n_{k+1}}$ to be the second best available good at the beginning of round $k$ in Round-Robin$(\bar{\bm{b}}^r_1, \mathbf{b}_{-1})$; 
recall that $\ell_i$ is the most valuable good with respect to $v^{r+1}_1$  that is allocated to some agent other than $1$  during round $i$.
By the definition of the $\hat{\lambda}_i$s, this observation implies
\[ v^{r+1}_1(\hat{\lambda}_i) \geq \max\big\{v^{r+1}_1(h_{n_{i+1}}), v^{r+1}_1(\ell_i)\big\}, \text{\ \ for all\ \ } r<i\leq k \,.\]
Given that $\bar{\bm{b}}^r_1$
trivially remains truthful from round $r+1$ with respect to $v^{r+1}_1$, the above maximum captures what is the mechanism's view of the second best good for agent 1 in round $i$ \emph{after} $\lambda_i$, for $i>r$, when the input is $(\bar{\bm{b}}^r_1, \mathbf{b}_{-1})$.
Plugging this bound into inequality \eqref{eq:sum_of_lambdas} along with $\hat{\lambda}_r = \lambda_r$, and rearranging the terms of round $r$, it yields
\[ \sum_{i= r+1}^{k} \big( v^{r+1}_1(\lambda_i) -  \max\big\{v^{r+1}_1(h_{n_{i+1}}), v^{r+1}_1(\ell_i)\big\} \big) =  v^{r+1}_1({\lambda}_r) - v^{r+1}_1(h_{n_r}) + \varepsilon \,,\]
as required. \renewcommand\qedsymbol{{\footnotesize $\boxdot$}}
\end{proof}

We can now define $v^{r}_1$ by appropriately changing some of the values of $v^{r+1}_1$. For the sake of readability, let \[\varepsilon_i = v^{r+1}_1(\lambda_i) -  \max\big\{v^{r+1}_1(h_{n_{i+1}}), v^{r+1}_1(\ell_i)\big\} \big)\] for $i>r$. 
Also, let $\alpha$ be such that $\varepsilon - \alpha\sum_{i=r+1}^{k}\varepsilon_i = \min\{\varepsilon, \delta\}/3$; recall that $\delta = \min_{g, h\in M, g\neq h}|v^{r+1}_1(g)-v^{r+1}_1(h)|$. Choosing such an $\alpha$ is always possible as a result of Claim \ref{claim:move_value}, because $f(\alpha) = \varepsilon - \alpha\sum_{i=r+1}^{k}\varepsilon_i$ with $\alpha\in(0,1)$ is a continuous function  with values in the interval  $(v^{r+1}_1(h_{n_r})-v^{r+1}_1({\lambda}_r), \varepsilon)$. We define:
\begin{itemize}[labelindent=20pt,leftmargin=*,itemsep=4pt,topsep=5pt]
	\item For each $r+1\le i \leq k$, we set $v_1^{r}(h_{n_{i}})=\max\big\{v^{r+1}_1(h_{n_{i+1}}), v^{r+1}_1(\ell_i)\big\} \big) + \alpha\cdot\varepsilon_i$.
	\item We also set $v_1^{r}(h_{n_{r}})=v_1^{r+1}(\lambda_{r}) + \varepsilon - \alpha\sum_{i=r+1}^{k}\varepsilon_i$.
	\item For any other good $g$, we set $v_1^{r}(g)=v_1^{r+1}(g)$.
\end{itemize}
We also define $\bm{b}^{r}_1$  from ${\bar{\bm{b}}}^r_1$ by replacing the bids with the corresponding $v^{r}_1$ values for all goods that are no better than $h_{n_r}$  in ${\bar{\succ}}^r_1$, i.e, we have
\[ g\, \not{\!\!{\bar{\succ}}}^r_1 \, h_{n_r} \Rightarrow \bm{b}^r_1(g) = v^{r}_1(g) \text{\ \ \ \ and\ \ \ \ } g\,\, {{\bar{\succ}}}^r_1 \,\,h_{n_r} \Rightarrow \bm{b}^r_1(g) = {\bar{\bm{b}}}^r_1(g)\,.\]

By this point, it should be straightforward to verify properties (i)-(iv). For (v), notice that any $\alpha$ resulting in $\varepsilon - \alpha\sum_{i=r+1}^{k}\varepsilon_i \in (0,\min\{\varepsilon, \delta\}/2)$ would work for consistently defining $\bm{b}^{r}_1$. If the above definition of $v^{r}_1$ happens to give one of the \emph{finitely many} values already in the range of $v^{r+1}_1$, then we may change $\alpha$ slightly to make all of the newly introduced values unique. 
\end{proof}

	\section{Towards EFX Equilibria: The Case of Two Agents}\label{sec:PR}

	As we saw, Round-Robin has PNE for every instance, and the corresponding allocations  are always \efo. The natural next question is \textit{can we have a similar guarantee for a stronger fairness notion?} In particular, we want to explore whether an analogous result is possible when we consider envy-freeness up to \textit{any} good.  When the agents are not strategic, it is known that \efx allocations exist when we have at most $3$ agents \cite{CaragiannisKMPS19, ChaGM20}. It should be noted that for the case of $3$ agents no polynomial time algorithm is known, and it is  unclear whether the constructive procedure of \citet{ChaGM20} has any PNE. For $n\ge 4$, the existence of \efx allocations remains a major open problem. Therefore, we turn our attention to the case of two agents.

	\subsection{A Mechanism with EFX Nash Equilibria}\label{sec:MEFXEQ}

	A polynomial-time algorithm that outputs \efx allocations when we have two agents is given by Plaut and Roughgarden \cite{PR18}. This is a modified \emph{cut-and-choose} algorithm where the cut (lines \ref{line:pr3}--\ref{line:pr5}) is produced using a variant of the \textit{envy-cycle-elimination} algorithm of Lipton et al.~\cite{LMMS04} on two copies of agent 1, and then  agent 2 ``chooses''  the best bundle among the two (line \ref{line:pr6}). We state it as mechanism Mod-Cut\&Choose below (recall the notation $\bm{b}_i(S)$ for $\sum_{h\in S} \bm{b}_i(h)$). We should  point out that this mechanism is not truthful, since there is no truthful mechanism for two agents that produces \efx (or \efo for that matter) allocations for more than four goods \cite{ABCM17}. 
    It is not obvious that the mechanism has PNE or that these are \efx, and even if that was the case, by Theorem \ref{thm:efx_to_mms_n=2}, there is no reason to expect that these PNE would guarantee more than $2\bmu_i / 3$ to each agent.
    Interestingly, we show that although not truthful, Mod-Cut\&Choose  always has at least one PNE for any instance, and all its equilibria are \mms and, by Theorem \ref{thm:mms_to_efx_n=2}, \efx. 
	
	\begin{algorithm}[ht]
		\caption{Mod-Cut\&Choose$(\bm{b}_1, \bm{b}_2)$ \cite{PR18} \hfill\small{ $\triangleright$ For $i\in \{1,2\}$, $\bm{b}_i = (b_{i1}, \ldots, b_{im})$ is the bid of agent $i$.}}
		\begin{algorithmic}[1]
			\State $(E_1, E_2) = (\emptyset,\emptyset)$
			\State $(h_1, h_2, \ldots, h_m)$ is $M$, sorted in decreasing order w.r.t.~$v_1$ \Comment{{\small Break ties lexicographically.}}
			\For{$i = 1, \dots, m$} \label{line:pr3}
                \State $j = \argmin_{k\in [2]} \bm{b}_1(E_k)$ \Comment{{\small Identify the worst bundle according to $\bm{b}_1$; break ties in favor of $E_1$.}}\vspace{3pt}\label{line:pr4}
				\State $E_j = E_j \cup \{h_i\}$ \Comment{{\small Add the next good to that bundle.}}\label{line:pr5}
			\EndFor
			\State $\ell = \argmax_{k\in [2]} \bm{b}_2(E_k)$ \Comment{{\small Identify the best bundle according to $\bm{b}_2$; break ties in favor of $E_1$.}}  \label{line:pr6}\vspace{3pt}			\State \textbf{Return}: {$\mathcal{A}=(M\setminus E_{\ell}, E_{\ell})$} \Comment{{\small Give this bundle to agent 2 and the remaining bundle to agent 1.}}
		\end{algorithmic}
		\label{alg:Cut+Choose}
	\end{algorithm}

	Seen as an algorithm, Mod-Cut\&Choose does not always produce $(5/6 + \varepsilon)$-\mms allocations for any $\varepsilon>0$, as it can be seen by the following simple instance with $N= \{1, 2\}$, and $M=\{g_1, \ldots, g_5\}$. For $i \in N$, let $v_i(g_j)=3$, if $j \in \{1,2\}$, and $v_i(g_j)=2$, if $j \in \{3,4,5\}$. It is easy to see that $\bmu_1 = \bmu_2 = 6$, but the allocation produced is $(\{g_1, g_3, g_5\}, \{g_2, g_4\})$, and thus agent 2 attains a value of $5$.

We begin with the following lemma on the ``cut'' part of Mod-Cut\&Choose, stating that agent 1 may create any desirable partition of the goods (up to the ordering of the two sets).
This is a necessary component of the proof of the main result of this section.

\begin{lemma}\label{lem:ECE}
	Let $(X_1, X_2)$ be a partition of $M$. Agent 1, by bidding accordingly, can force Mod-Cut\&Choose to construct $E_1, E_2$ in lines \ref{line:pr3}--\ref{line:pr5}, such that $\{E_1, E_2\} = \{X_1, X_2\}$.
\end{lemma}
\begin{proof}
    We consider different cases depending on the cardinality of the sets $X_1, X_2$. Each case describes a bid that agent 1 can report in order to create the desired partition $(X_1, X_2)$ or its permutation $(X_2, X_1)$. Note that only the first case is relevant when $m=1$, and only the first two cases are relevant when $2\le m\leq 3$. 
	\smallskip
	
	\noindent\ul{Case 1 (one set has all the goods).} 
	Agent 1 declares zero value for all the goods. According to these values, $j$ in line \ref{line:pr4} is always $1$, so every good goes to $E_1$, and we have the desired partition.	
	\smallskip
	
	\noindent\ul{Case 2 (one set has $m-1$ goods).}
	Agent 1 declares value $1$ for the good that is contained in the set with cardinality $1$, and for every good that is contained in the set with cardinality $m-1$ she declares a value equal to $\frac{1}{m-1}$. The first good is added in $E_1$, so $E_2$ is going to get chosen next. Actually, according to these values, $E_2$ must get all the remaining goods. Thus, the desired partition is produced.
	\smallskip
	
	\noindent\ul{Case 3 (the two sets have cardinalities $k \geq 2$ and $m-k \geq 2$).} 
	Agent 1 declares a value of $1$ for one of the goods that are contained in the set with cardinality $k$. For every good that is contained in the set with cardinality $m-k$ she declares a value equal to $\frac{1+\varepsilon}{m-k}$, where $0< \varepsilon <\frac{1}{m-k-1}$. Finally, for the rest of the goods that are contained in the set with cardinality $k$ she declares a value of $\frac{\varepsilon}{k-1}$. $E_1$ gets the first good, so it appears to be more valuable than $E_2$. According to these values, $E_2$ ceases to appear to be the worst of the two when it gets every good of the set with cardinality $m-k$. This is the point where $E_1$ becomes worse than $E_2$, and continues to be worse until it contains every good of the set of cardinality $k$. Thus, again,  the desired partition is produced. 
\end{proof}

In particular, agent 1 can force the mechanism to construct $E_1, E_2$, such that $\min\{v_1(E_1), v_1(E_2)\} = \bmu_1$. Such a pair $(E_1, E_2)$ is called a \textit{$\bmu_1$-\,partition}. At least one $\bmu_1$-\,partition exists, by the definition of $\bmu_1$.
	
	\begin{corollary}\label{cor:ECEMMS}
		Agent 1 can force Mod-Cut\&Choose to construct a $\bmu_1$-\,partition in lines \ref{line:pr3}--\ref{line:pr5}.
	\end{corollary}
	
We can now proceed to the main theorem of this section on the existence and fairness properties of the PNE of Mod-Cut\&Choose.
	
	\begin{theorem}\label{thm:mms+efx_PNE}
		For any instance $\mathcal{I}=(\{1,2\},M,\mathbf{v})$, the Mod-Cut\&Choose mechanism has at least one PNE. Moreover, every PNE of the mechanism is \mms and \efx with respect to the valuation functions $v_1,v_2$.
	\end{theorem}
	
	\begin{proof}
	Given a partition $\mathcal{X} = (X_1,X_2)$ we are going to slightly abuse the notation---as we do in our pseudocode---and consider $\argmin_{X\in\mathcal{X}} v_2(X)$ to be a single set in $\mathcal{X}$ rather than a subset of $\{X_1,X_2\}$. To do so, we assume that ties are  broken in favor of the highest indexed set (here $X_2$) and tie-breaking is applied by the $\argmin$ operator.
	
	We will define a profile $(\bm{b}_1,\bm{b}_2)$ and show that it is a PNE.
	First, let $\bm{b}_2 = (v_2(g_1), v_2(g_2), \ldots, v_2(g_m))$ be the truthful bid of agent 2. Next $\bm{b}_1$ is the bid vector (as defined within the proof of Lemma \ref{lem:ECE}) that results in Mod-Cut\&Choose constructing a partition in 
	\[\argmax_{\mathcal{X} \in \Pi_2(M)} v_1 \big( \argmin_{X\in\mathcal{X}} v_2(X) \big) \,.\]
	To see that there exists such $\bm{b}_1$, notice that the set $\Pi_2(M)$ of all possible partitions  is finite and, by Lemma \ref{lem:ECE}, every possible partition can be produced by Mod-Cut\&Choose given the appropriate bid vector of agent 1. So, agent 1 forces the partition that maximizes, according to $v_1$, the value of the least desirable bundle according to $v_2$. 
	Now it is easy to see that given the bidding strategy of agent 2, i.e., playing truthfully, there is no deviation for agent 1 that is profitable (by definition). Moreover, agent 2 gets the best of the two bundles according to her  valuation function (regardless of the partition, truth telling is a dominant strategy for her), thus there is no profitable deviation for her either. Therefore, $(\bm{b}_1,\bm{b}_2)$ is a PNE for $\mathcal{I}$.
		
	Regarding the second part of the statement, suppose for a contradiction that there is a PNE $\mathbf{b}$, where an agent $i$ does not achieve her $\bmu_i$ in the allocation returned by Mod-Cut\&Choose$(\mathbf{b})$. If this agent is agent 1, then according to Corollary \ref{cor:ECEMMS}, there is a bid vector $\bm{b}'_1$ she can report, so that the algorithm will produce a $\bmu_1$-\,partition. By deviating to $\bm{b}'_1$, regardless of the set given to  agent 2, agent 1 will end up with a bundle she values at least $\bmu_1$. As this would be a strict improvement over what she currently gets, it would contradict the fact that $\mathbf{b}$ is a PNE. So, it must be the case where agent 2 gets a bundle she values strictly less than $\bmu_2$. Notice that, regardless of the partition which Mod-Cut\&Choose to constructs in lines \ref{line:pr3}--\ref{line:pr5}, by declaring her truthful bid, agent 2 gets a bundle of value at least $v_2(M)/2$. By Definition \ref{def:mmshare}, it is immediate to see that this value is at least $\bmu_2$, i.e., deviating to her truthful bid is a strict improvement over what she currently gets by Mod-Cut\&Choose$(\mathbf{b})$, which is a contradiction.
	
	It remains to show that the allocation returned by Mod-Cut\&Choose$(\mathbf{b})$ is also \efx. However, since here $n=2$, this directly follows from Theorem \ref{thm:mms_to_efx_n=2}. 
	\end{proof}

\subsection{The Enhanced Fairness of EFX Nash Equilibria} \label{sec:efx_to_mms_on_pne}
	
As it was discussed in Section \ref{sec:MEFXEQ}, it is surprising that  the \efx equilibria of Mod-Cut\&Choose impose stronger fairness guarantees compared to generic \efx allocations or even \efx allocations produced by Mod-Cut\&Choose itself in the non-strategic setting. 
In this section we explore whether something similar holds for \textit{every} mechanism with \efx equilibria. Specifically, we consider the (obviously non-empty) class of mechanisms
that have PNE for every instance and these equilibria always lead to \efx allocations. Our goal is to determine if these allocations have  better fairness guarantees (with respect to the underlying true valuation functions) than \efx allocations in general. To this end, we start by examining instances of two agents and  $4$ goods and we prove that for every mechanism of this class, all allocations at a PNE are \mms allocations.  The reason we start from this restricted set of instances is that it already provides a clear separation with the non-strategic setting. Recall from Theorem \ref{thm:efx_to_mms_n=2} that there are instances with just $4$ goods where an \efx allocation may not be a $\big(\frac{2}{3} +\delta\big)$-\mms allocation, for any $\delta>0$.

We begin by showing the following lemma which regards some very simple cases of such instances, and then we proceed to the proof of the statement. Recall that $\bmu_i$ denotes the maximin share of agent $i$ (see Definition \ref{def:mmshare}).
\begin{lemma}\label{lem:3EFX=MMS}
	Consider an instance with $2$ agents and $4$ goods. If  agent $i\in[2]$ has strictly positive value for three or less goods, then in every allocation which is \efx from her point of view, agent $i$ has value at least $\bmu_i$. 
\end{lemma}
	
\begin{proof}
	Suppose agent $i$ has positive value for at most three goods.  The statement is trivial when there is \textit{at most one} positively valued good as in this case $\bmu_i  = 0$ and agent $i$ always gets $\bmu_i$ no matter the bundle that she gets. When she has a positive value for \textit{two} goods, in order to consider the allocation as \efx she must get at least one of them. In this case she also achieves her $\bmu_i$ as it is equal to the smaller of the two positive values.  Finally, suppose agent $i$ has positive value for \textit{three} goods. Notice that $\bmu_i$ in this case is either equal to the largest of the three values or to the sum of the two smallest values; whichever is smaller. So, if agent $i$ gets two goods, then she always derives a value of at least gets $\bmu_i$. If she gets just one good, then this good must have the highest value, otherwise the she would not consider the allocation as \efx.  So, in this case too, she gets value at least $\bmu_i$.
\end{proof}

We are now ready for the general result. 
	\begin{theorem}\label{thm:4EFX=MMS}
		Let $\mathcal{M}$ be a mechanism that has PNE for any instance 
		$(\{1, 2\},M,(v_1, v_2))$ with $|M|=4$,
		and all these equilibria lead to \efx allocations with respect to $v_1, v_2$. Then each such \efx allocation is also an \mms allocation. 
	\end{theorem}
	
	\begin{proof}
		Suppose for contradiction that this is not the case. This means that there exists a valuation instance $\mathbf{v}=(v_1, v_2)$, for which there is a PNE $\mathbf{b}=(\bm{b}_1, \bm{b}_2)$ that produces an \efx allocation $(A_1, A_2)$, where, without loss of generality,  $v_1(A_1)<\bmu_1$. Rename, if necessary, the goods to $\{h_1, h_2, h_3, h_4\}$, so that $v_1(h_1)\geq v_1(h_2) \geq v_1(h_3)\geq v_1(h_4) > 0$, where the last inequality follows from Lemma \ref{lem:3EFX=MMS}. 
        The following lemma, established within the proof of Theorem 5.1 of \citet{ABM16} (also Lemma 5.3 in \citep{ANM2019}), will reduce and simplify the possible cases we need to consider.
        \begin{lemma}[Follows from the proof of Theorem 5.1 of \citep{ABM16}]\label{lem:mms_4_goods}
        For $N$, $M$ and $v_1$ as above, we have $v_i(\{h_1, h_4\}) \ge \bmu_1$ and $\max\left\{v_i(\{h_1\}), v_i(\{h_2, h_3\})\right\} \ge \bmu_1$.
        \end{lemma}
        Given Lemma \ref{lem:mms_4_goods}, the bundle $A_1$ must be either a singleton or one of $\{h_2, h_3\}, \{h_2, h_4\}, \{h_3, h_4\}$. 
        \smallskip

        \noindent\ul{Case 1 ($|A_1|=1$).} Since $(A_1, A_2)$ is an \efx allocation and all goods have positive value according to $v_1$, it is easy to see that $A_1 = \{h_1\}$. Then, again because we have an \efx allocation, $v_1(h_1) \ge v_1(\{h_2, h_3\})$. The latter implies $v_1(A_1) \ge \bmu_1$, by the second inequality of Lemma \ref{lem:mms_4_goods}.  
		\smallskip

		\noindent\ul{Case 2 ($A_1=\{h_2, h_3\}$).} Since $(A_1, A_2)$ is an \efx allocation, we have  $v_1(\{h_2, h_3\})\geq v_1(h_1)$. 
        Like in Case 1, this implies the contradiction $v_1(A_1) \ge \bmu_1$, by the second inequality of Lemma \ref{lem:mms_4_goods}.
		  \smallskip

        \noindent\ul{Case 3 ($A_1=\{h_2, h_4\}$ or $A_1=\{h_3, h_4\}$).} So far we have not use the fact that $\mathbf{b}=(\bm{b}_1, \bm{b}_2)$ is a PNE for the valuation profile $ \mathbf{v}=(v_1, v_2)$. Consider a different valuation profile $\mathbf{v}^*=(v^*, v^*)$, where the agents have identical values over the goods. The valuation function $v^*$ is defined as follows:
		\[ 
		v^*(h_j)= \left\{
		\begin{array}{lll}
			1.2 & j=1 \\
			1 &j \in \{2,3\} \\
			0.1 & j=4  \\
		\end{array} 
		\right. 
		\]
		It is easy to see that for $\mathbf{v}^*$ there are only two \efx allocations, namely $(\{h_1, h_4\},\{h_2, h_3\})$ and its symmetric $(\{h_2, h_3\}, \{h_1, h_4\})$. According to our assumption, there must be a bid vector $\mathbf{ b}^*=(\bm{b}^*_1, \bm{b}^*_2)$ that is a PNE of $\mathcal{M}$ for this valuation profile, and since we require the PNE to be also \efx, $\mathcal{M}(\mathbf{ b}^*)$ must be one of these allocations.  Moreover, observe that the value that agent 2 derives in these allocations is at most $2$.
  Let us examine what each agent can get if agent 1 deviates from $\mathbf{ b}$ to $\mathbf{ b}'=(\bm{b}^*_1, \bm{b}_2)$:
		\begin{itemize}
        [labelindent=10pt,leftmargin=*,itemsep=2pt,topsep=3pt]
			\item In case the bundle of agent 1 is a singleton, then agent 2 gets a bundle of cardinality $3$. This contradicts the fact that $\mathbf{ b}^*=(\bm{b}^*_1, \bm{b}^*_2)$ is a PNE for the valuation profile $\mathbf{ v}^*=(v^*, v^*)$, as any such set gives agent 2 a value of at least $2.1$.
			\item In case the bundle of agent 1 has cardinality 3, this contradicts the fact that $\mathbf{ b}=(\bm{b}_1, \bm{b}_2)$ is a PNE for the valuation profile $\mathbf {v}=(v_1, v_2)$, as the least valuable such set is $\{h_2, h_3, h_4\}$ and it has strictly more  value than $v_1(A_1)$, since $v_1(h_j)>0$ for every $j\in [4]$.
			\item In case the bundle of agent 1 is one of $\{h_1, h_2\}$, $\{h_1, h_3\}$, $\{h_1, h_4\}$, or $\{h_2, h_3\}$, then this implies that $A_1$ has value at least equal to the value of one of these bundles. By using Lemma \ref{lem:mms_4_goods} as above, we get the contradiction $v_1(A_1)\geq\bmu_1$.
			\item In case the bundle of agent 1 is one of $\{h_2, h_4\}$ or $\{h_3, h_4\}$, then agent 2 gets either $\{h_1, h_2\}$ or $\{h_1, h_3\}$. This contradicts the fact that $\mathbf{ b}^*=(\bm{b}^*_1, \bm{b}^*_2)$ is a PNE for the valuation profile $\mathbf{ v}^*=(v^*, v^*)$, as any such set gives agent 2 a value of at least $2.2$.
		\end{itemize}
        \smallskip
		
		Since every possible case leads to a contradiction, we conclude that every allocation corresponding to a PNE of $\mathcal{M}$ guarantees to each agent her maximin share. 
	\end{proof}

The proof of Theorem \ref{thm:4EFX=MMS}
relies on extensive case analysis, part of which is hidden within Lemma \ref{lem:mms_4_goods}. Each case assuming that the allocation is \efx but not \mms eventually contradicts the fact that the current profile is a PNE. When we consider instances with $5$ or more goods, this approach is not fruitful anymore. The reason is not solely the increased number of cases one has to handle, but rather the fact that now some of the cases do not seem to lead to a contradiction at all.

Although we suspect that the theorem is no longer true for more than $4$ goods, we are able prove a somewhat weaker property that still separates the \efx allocations in PNE from generic \efx allocations in the non-strategic setting. 
In particular, for general mechanisms that have PNE for every instance and these equilibria are always \efx, we show that the corresponding  allocations always guarantee an approximation to \mms that is \textit{strictly better} than $2/3$.

\begin{theorem}\label{thm:NEFX_APMMS}
	Let $\mathcal{M}$ be a mechanism that has PNE for any instance $(\{1,2\},M,(v_1, v_2))$, and all these equilibria lead to \efx allocations with respect to $v_1, v_2$. Then each such \efx allocation is also an $\alpha$-\mms allocation for some $\alpha> 2/3$.
\end{theorem}
	
\begin{proof}
	Suppose for a contradiction that this is not the case. This means that there exists such a mechanism $\mathcal{M}$ and an instance $(\{1,2\},M,(v_1,v_2))$, for which there is a PNE $\mathbf{b}=(\bm{b}_1, \bm{b}_2)$ that results in an \efx allocation $\mathcal{A} = (A_1, A_2)$, where $v_i(A_i)\le2\bmu_i /3$ for at least one $i\in [2]$. Without loss of generality, assume $v_1(A_1)\le2\bmu_1 /3$ and notice that this means that $v_1(A_1) = 2\bmu_i /3$, as $v_1(A_1)$ cannot be smaller than $2\bmu_i /3$, by Theorem \ref{thm:efx_to_mms_n=2}. 
	This implies that $v_1(A_2)\geq 4\bmu_i /3$, since  $v_1(M)\ge 2\bmu_1$ by Definition \ref{def:mmshare}. 
		
	Initially, we will restrict the number of the goods with positive value (according to $v_1$) in $A_2$. Let $S\subseteq A_2$ be the set of such goods, i.e, $S=\{g\in A_2 \,|\, v_1(g)>0\}$.  Let $|S| = k$ and notice that $k$ cannot be $0$ or $1$ since otherwise $v_1(A_1) \geq \bmu_1$. 
	Finally,  let $x\in \argmin_{g\in S} v_1(g)$ be a minimum valued good for agent 1 in $S$. We have
\[ \frac{2}{3}\bmu_1  = v_1(A_1) \ge v_1(S\setminus\{x\}) \geq v_1(S) - \frac{v_1(S)}{k} = \frac{(k-1)}{k}v_1(A_2)\geq \frac{(k-1)}{k} \frac{4}{3}\bmu_1\,,  \]
    where the first inequality follows from $(A_1, A_2)$ being \efx. Given our observation that $k\le 2$, the above implies that $k=2$. 
	Name $h_1$ and $h_2$ the goods of $S$, and observe that if $v_1(A_2)=v_1(\{h_1, h_2\})> 4\bmu_1 /3$, then $(A_1, A_2)$ cannot be \efx from the perspective of agent 1. Thus, we get that $v_1(A_2)=  4\bmu_1 /3$, which in conjunction with \efx implies  $v_1(h_1)=v_1(h_2)= 2\bmu_1 /3$. 
		
	Next we argue that $A_1$ contains at least $2$ goods that have positive value for agent 1. 	
	Indeed, if all the goods in $A_1$ had zero value, then we would have $v_1(A_1)=0< 2 \bmu_1 / 3$ as $A_2$ contains two positively valued goods, while if there was just one positively valued good in $A_1$, this would imply that only three goods have positive value for agent 1, and each one of them has value $2 \bmu_1 / 3$. The latter would make the existence of a $\bmu_1$-\,partition impossible, which is a contradiction. 
	So, since there are at least two positively valued goods in $A_1$ for agent 1, we arbitrarily choose two of them, and we name them $h_3$ and $h_4$. 
	We arbitrarily name the remaining goods $h_5, h_6,\ldots, h_m$.
		
	Consider now a different valuation instance $\mathbf{v}^*=(v^*, v^*)$ where the agents have identical values over the goods. The valuation function is defined as 
		\[ 
		v^*(h_j)= \left\{
		\begin{array}{lll}
			1.2 & j=1 \\
			1 &j \in \{2,3\} \\
			\varepsilon & j \in \{4,\ldots , m\}  \\
		\end{array} 
		\right. 
		\]
		where $\varepsilon>0$ and $(m-3)\cdot \varepsilon < 0.2$. It is easy to see that for this valuation instance there are only two \efx allocations, namely, $\mathcal{X} = (\{h_1, h_4,\ldots, h_m\},\{h_2, h_3\})$, and its symmetric $\mathcal{Y} = (\{h_2, h_3\}, \{h_1,h_4,\ldots,h_m\})$. According to our assumption, there must be a bidding vector $\mathbf{b}^*=(\bm{b}^*_1, \bm{b}^*_2)$ that is a PNE of 
		$\mathcal{M}$ for the instance $(\{1,2\},M,\mathbf{v}^*)$, and since all PNE of $\mathcal{M}$ are also EFX,  $\mathcal{M}(\mathbf{b}^*)$ must output one of $\mathcal{X}$ and $\mathcal{Y}$.  Moreover, observe that the value agent 2 receives (with respect to $V^*$) in these allocations is $2$ and $1.2+ (m-3) \varepsilon < 1.4$ respectively.

		For now assume that $\bm{b}_1 \neq \bm{b}^*_1$ and $\bm{b}_2 \neq \bm{b}^*_2$. 
		We will show that, in this case, running $\mathcal{M}$ with input $\mathbf{b}'=(\bm{b}^*_1, \bm{b}_2)$ results to agent 2 receiving a bundle of value strictly better than $2$ according to $v^*$. This  contradicts
		the fact that $\mathbf{b}^*=(\bm{b}^*_1, \bm{b}^*_2)$ is a PNE for $\mathbf{v}^*=(v^*, v^*)$. 
Recall that $\mathbf{b}=(\bm{b}_1, \bm{b}_2)$ is a PNE for $\mathbf{v}=(v_1, v_2)$,  that $v_1(h_1)=v_1(h_2)= v_1(A_1)= 2 \bmu_1 / 3$, and that $v_1(h_3), v_1(h_4)$ are strictly positive. 
		So, let us examine what each agent may get if agent 1 deviates from $\mathbf{b}$ to $\mathbf{b}'=(\bm{b}^*_1, \bm{b}_2)$:
		\begin{itemize}[labelindent=10pt,leftmargin=*,itemsep=2pt,topsep=3pt]
			\item In case the bundle of agent 1 contains good $h_1$, it cannot contain any good from $\{h_2,h_3,h_4\}$; otherwise $\mathbf{b}=(\bm{b}_1, \bm{b}_2)$ would not be a PNE for  $\mathbf{v}=(v_1, v_2)$. Thus, $\{h_2,h_3,h_4\}$ is part of the bundle of agent 2.
			\item In case the bundle of agent 1 contains good $h_2$, it cannot contain any good from $\{h_1,h_3,h_4\}$; otherwise $\mathbf{b}=(\bm{b}_1, \bm{b}_2)$ would not be a PNE for $\mathbf{v}=(v_1, v_2)$. Thus, $\{h_1,h_3,h_4\}$ is part the bundle of agent 2.
			\item In case the bundle of agent 1 does not contain any of $h_1$ and $h_2$, then it is possible for her to get any subset $T\subseteq \{h_3, h_4, \ldots , h_m\}$. However, $\{h_1,h_2\}$ is part the bundle of agent 2.
		\end{itemize}
		Thus, in the allocation returned by $\mathcal{M}(\mathbf{b}')$, agent 2 gets a bundle that contains $\{h_2,h_3,h_4\}$ or  $\{h_1,h_3,h_4\}$ or $\{h_1,h_2\}$. Consider the value of these sets according to $v^*$:
		\[v^*(\{h_2,h_3,h_4\})=2+\varepsilon\,, \qquad 
		  v^*(\{h_1,h_3,h_4\})=2.2+\varepsilon\,, \qquad
		  v^*(\{h_1,h_2\})=2.2 \,.\]
		That is, in every single case the value agent 2 derives under $\mathbf{v}^*=(v^*, v^*)$ when the profile $\mathbf{b}'=(\bm{b}^*_1, \bm{b}_2)$ is played is strictly better than $2$. However, $2$ is the maximum possible value that agent 2 could derive under $\mathbf{v}^*$ when the profile $\mathbf{b}^*$ is played. This contradicts the fact that $\mathbf{b}^*$ is a PNE for $\mathbf{v}^*$, as $\bm{b}_2$ is a profitable deviation for agent 2. 
		
		The remaining corner cases are straightforward to deal with. To begin with, it is not possible to have $\bm{b}_1 =  \bm{b}^*_1$ and $\bm{b}_2 = \bm{b}^*_2$, as $\mathcal{X}\neq \mathcal{A}$ and $\mathcal{Y}\neq \mathcal{A}$.
		
		Next, assume that $\bm{b}_1= \bm{b}^*_1$ and $\bm{b}_2 \neq \bm{b}^*_2$.  This directly contradicts the fact that $\mathbf{b}^*$ is a PNE for $\mathbf{v}^*=(v^*, v^*)$. To see this, starting from $\mathbf{b}^*$ let agent 2 deviate to $\bm{b}_2$. She then gets $A_2$ which contains $h_1, h_2$ and has value for her $v^*(A_2)\ge 2.2 > 2$.
		
		Finally, assume that $\bm{b}_1\neq \bm{b}^*_1$ and $\bm{b}_2 = \bm{b}^*_2$. This directly contradicts the fact that
		$\mathbf{b}$ is a PNE for $\mathbf{v}=(v_1, v_2)$. To see this, starting from $\mathbf{b}$ let agent 1 deviate to $\bm{b}^*_1$. She either gets
		$\{h_1, h_4,\ldots , h_m\}$ of value at least $v_1(h_1)+v_1(h_4)> 2\bmu_1 /3 = v_1(A_1)$ or she gets $\{h_2, h_3\}$ of value $v_1(h_2)+v_1(h_3)> 2\bmu_1 /3 = v_1(A_1)$.
		
		Since every possible case leads to a contradiction, we conclude that every allocation that corresponds to a PNE of a mechanism in the class of interest, guarantees to each agent $i$ value that is strictly better than $2\bmu_i /3$, for $i\in[2]$. 
	\end{proof}

\section{Discussion}

In this work we studied the problem of fairly allocating a set of indivisible goods, to a set of strategic agents. Somewhat surprising---given the existing strong impossibilities for truthful mechanisms---our results are mostly positive. In particular, we showed that there exist mechanisms that have PNE for every instance, and at the same time the allocations that correspond to PNE have strong fairness guarantees with respect to the true valuation functions. 

We believe that there are several interesting directions for future work that follow our research agenda.
For instance, it would be interesting to explore how algorithms that compute \efo allocations for richer valuation function domains (e.g., the Envy-Cycle-Elimination algorithm \cite{LMMS04}) behave in the strategic setting we study in this work. 
Here the question is twofold. On one hand, it is unclear whether such algorithms have Nash equilibria---pure or mixed---for every valuation instance; on the other, it would be important to determine if they maintain their fairness properties at their equilibria or not. 
The existence of PNE or MNE for algorithms that compute approximate \mms allocation is on a similar direction and, as we mentioned in Section \ref{sec:prelims}, in this case we get the ex-post or ex-ante \mms guarantee on the equilibria for free.

Theorems \ref{thm:4EFX=MMS} and \ref{thm:NEFX_APMMS} leave an open question on the \mms guarantee that the equilibria of mechanisms that always have PNE and these are \efx. Although we suspect that the corresponding allocations are not always \mms,  such a result would immediately imply that for every such mechanism which runs in polynomial time, finding a best response of an agent is a computationally hard problem. 
Going beyond the case of two agents here seems to be a highly nontrivial problem as it is not very plausible that the current state of the art for the non-strategic setting could be analysed under incentives.

Finally, although we did not really focus on  complexity questions, it is clear that 
computing  best responses is generally hard. However, when they are not, for instance when the number of agents in Round-Robin is fixed \cite{XiaoL20}, we would like to know if best response dynamics always converge to a PNE  or there  might be cyclic behavior (as it happens with better response dynamics \cite{GW17}).

\section*{Acknowledgments.}
This work was supported by the ERC Advanced 
 Grant 788893 AMDROMA ``Algorithmic and Mechanism Design Research in 
 Online Markets'', the MIUR PRIN project ALGADIMAR ``Algorithms, Games, 
 and Digital Markets'', the FAIR (Future Artificial Intelligence Research) project, funded by the NextGenerationEU program within the PNRR-PE-AI scheme (M4C2, investment 1.3, line on Artificial Intelligence), PNRR MUR project  IR0000013-SoBigData.it, the NWO Veni project No.~VI.Veni.192.153, and the project MIS 5154714 of the National Recovery and Resilience Plan Greece 2.0 funded by the European Union under the NextGenerationEU Program.

%
%
%

\appendix
\section{Dealing With Ties Among the Values}\label{app:strict-vs-ties}

We begin with a lemma that is used twice: once to show that Round-Robin has PNE for any instance (Theorem \ref{thm:bluff}), and then again in the complete proof of Theorem \ref{thm:best_response}. Both proofs are presented in this appendix.

\begin{lemma}\label{lem:ties}
For any fair division instance $\mathcal{I}=(N,M,\mathbf{v})$ and any agent $i\in N$, there exists a valuation function $v'_i$ with the following properties:
\begin{itemize}[labelindent=10pt,leftmargin=*,itemsep=3pt,topsep=3pt]
    \item $v'_i$ induces a strict preference ranking over $M$, which is consistent with the preference ranking induced by $v_i$;
    \item if a bid vector $\bm{b}_i$ is a best response of agent $i$ with respect to $v_i$ to the (fixed) bid vectors $\mathbf{b}_{-i}$ of all other players in Round-Robin, then $\bm{b}_i$ is still a best response to $\mathbf{b}_{-i}$ with respect to $v'_i$;
    \item $v_i(T)\le v'_i(T)\le v_i(T) + \varepsilon / 3$, for any $T\subseteq M$, where $\varepsilon$ is the smallest positive difference between the values of two goods with respect to $v_1$, i.e., $\varepsilon = \min\{|v_i(g)-v_i(h)| \,:\, g,h\in M,\,v_i(g)\neq v_i(h)\}$, or,  if there is no positive difference, $\varepsilon = 1$.
\end{itemize}
\end{lemma}

\begin{proof}
If $v_i$ already induces a strict preference ranking over the  goods, then clearly $v'_i=v_i$ has all these properties. So, suppose that there are goods with exactly the same  $v_1$ value and let $S = \{g\in M \,:\, \exists h\in M \text{\ such that\ } h\neq g \text{\ and\ } v_i(h)=v_i(g)\}$ be the set of all goods that do not have a unique $v_1$ value. Also, for $\mathbf{b} = (\bm{b}_1, \bm{b}_2, \ldots, \bm{b}_n)$ as in the second bullet of the statement, let $(A_1, A_2,\ldots, A_n)$ be the allocation returned by Round-Robin$(\mathbf{b})$. Then we define $v'_i$ on $M=\{g_1,\ldots,g_m\}$ as follows
	\[ 
		v'_i(g_j)= \left\{
		\begin{array}{lll}
			v_i(g_j) + \displaystyle\frac{j\cdot \varepsilon}{3m^2}\, , &\ \  \text{ if } g_j \in S\cap A_i \vspace{4pt}\\
			v_i(g_j) + \displaystyle\frac{j\cdot \varepsilon}{6m^5}\, , &\ \  \text{ if } g_j \in S\setminus A_i \vspace{4pt}\\
			v_i(g_j)\, , &\ \  \text{ if } g_j \in M\setminus S 
		\end{array} 
		\right. 
		\]

It is straightforward to verify that ties are broken without introducing any new ties and without violating the preference ranking induced by $v_i$. Also, the added quantities sum up to a value smaller than $\varepsilon/3$. So the first and third properties hold for $v'_i$. To see that $\bm{b}_i$ is still a best response to $\mathbf{b}_{-1}$ with respect to $v'_i$, suppose for a contradiction that this is not the case. That is, there is some bid vector $\bm{b}'_i$, such that in the allocation $(A'_1, A'_2,\ldots, A'_n)$  returned by Round-Robin$(\bm{b}'_i,\mathbf{b}_{-i})$ we have $v'_i(A'_i)>v'_i(A_i)$. The latter implies that $A'_i \neq A_i$. Given that $v_i(A'_i) \le v_i(A_i)$, we distinguish two cases. 

First, suppose $v_i(A'_i) < v_i(A_i)$. By the definition of $\varepsilon$ we have $v_i(A'_i) \le v_i(A_i) - \varepsilon$. This, however, implies $v'_i(A'_i) \le v_i(A'_i) + \varepsilon / 3 \le v_i(A_i) - 2\varepsilon / 3 \le v'_i(A_i) - 2\varepsilon / 3$, which contradicts  $v'_i(A'_i)>v'_i(A_i)$.

So, it must be the case that $v_i(A'_i) = v_i(A_i)$. Then the difference $v'_i(A'_i)-v'_i(A_i)$ must be due to the small terms we added to the values of some goods. Note that if $A_i\not\subseteq A'_i$, then the  value added to $v_i(A'_i)$ is at most the total value added to $S\setminus A_i$ plus \textit{almost} the total value added to $S\cap A_i$, as we should exclude at least $\frac{\varepsilon}{3m^2}$. But as $\sum_{j=1}^m \frac{j\cdot \varepsilon}{6m^5} < \frac{\varepsilon}{3m^2}$, we have that this difference $v'_i(A'_i)-v'_i(A_i)$ must be negative, and this again contradicts  $v'_i(A'_i)>v'_i(A_i)$. So it must be the case where $A_i \subseteq A'_i$. However, independently of the bid profile, we know that agent $i$ will receive \textit{exactly} the same number of goods in the two executions of Round-Robin (i.e., either $\lceil m/n \rceil$  in both or $\lfloor m/n \rfloor$ in both). But then $A_i = A'_i$  again contradicting $v'_i(A'_i)>v'_i(A_i)$.

We conclude that $\bm{b}_i$ is still a best response to $\mathbf{b}_{-1}$ with respect to $v'_i$. 
\end{proof}

Recall that \citet{GW17} showed is that as long as all the valuation functions in an instance induce strict preference rankings and all the values are positive, then there is a way to construct PNE. In the terminology of \citep{GW17} these are all the bid profiles that are consistent with the so-called \emph{bluff profile}  defined therein. Here we do not need to define what the bluff profile is explicitly. We are going to use the following result which essentially is a corollary of \citep{GW17}.

\begin{theorem}[Follows from \citep{GW17}]\label{thm:bluff}
   For any instance $\mathcal{I}=(N,M,\mathbf{v})$, where all goods have positive values for all agents and all the valuation functions induce strict preference rankings, Round-Robin has at least one PNE.
\end{theorem}

Using Theorem \ref{thm:bluff} and Lemma \ref{lem:ties}, we will show that Round-Robin has PNE in every single instance with additive valuation functions.

\begin{theorem}\label{thm:general_bluff}
   For any instance $\mathcal{I}=(N,M,\mathbf{v})$ Round-Robin has at least one PNE.
\end{theorem}

\begin{proof}
   For each one of $v_1, v_2, \ldots, v_n$ we apply Lemma \ref{lem:ties} to get $\mathbf{v}' = (v'_1, v'_2, \ldots, v'_n)$. When we apply it for $v_i$, let $\varepsilon_i$ be the corresponding constant of the third bullet of the lemma; the second bullet of the lemma is irrelevant here.
   Clearly, for all $i\in N$, $v'_i$ induces a strict preference ranking, so for Theorem  \ref{thm:bluff} to apply we only need that all values are positive. This may not always be the case. If $v_i$ assigned value $0$ to multiple goods, then all the $0$ are taken care of during the definition of $v'_i$. If, however there was a single good $g$ such that $v_i(g)=0$, then $v'_i(g)=0$ as well. We can resolve this by setting $v'_i(g)=\varepsilon_i / 3$. This does not affect the induced preference ranking of $v'_i$, while the property of the third bullet of Lemma \ref{lem:ties} becomes $v_i(T)\le v'_i(T)\le v_i(T) + 2 \varepsilon / 3$ instead.
   
   Now we may apply Theorem \ref{thm:bluff}. Thus, for the instance $\mathcal{I}'=(N,M,\mathbf{v}')$ Round-Robin has at least one PNE; suppose the profile $\mathbf{d} = (\bm{d}_1, \bm{d}_2, \ldots, \bm{d}_n)$ is such a PNE and let $(A_1, A_2, \ldots, A_m)$ be the allocation returned by Round-Robin$(\mathbf{d})$. We claim that $\mathbf{d}$ is also an equilibrium of the original instance $\mathcal{I}$. 
   
   Suppose it is not, for a contradiction. This means that in $\mathcal{I}$ there is an agent, say agent $k$, who can deviate to a bid profile $\bm{b}_k$, so that the allocation returned by Round-Robin$(\bm{b}_k, \mathbf{d}_{-k})$ is $(B_1, B_2, \ldots, B_m)$ with $v_k(B_k) > v_k(A_k)$.
   By the definition of $\varepsilon_k$, we have $v_k(A_k) \le v_k(B_k) - \varepsilon_k$. 
   This implies $v'_k(A_k) \le v_k(A_k) + 2\varepsilon_k / 3 \le v_k(B_k) - \varepsilon_k / 3 \le v'_k(B_k) - \varepsilon / 3 < v'_k(B_k)$, which contradicts  the fact that $\mathbf{d}$ is a PNE in $\mathcal{I}'$. 
   Therefore, $\mathbf{d}$ is a PNE in $\mathcal{I}$ as well. 
\end{proof}

Finally, we can present a complete proof of Theorem \ref{thm:best_response}, without any assumptions on the valuation function of agent 1.

\paragraph*{Complete Proof of Theorem \ref{thm:best_response}.}
Consider an arbitrary instance $\mathcal{I}=(N,M,\mathbf{v})$ and assume that the input of Round-Robin is $\,\mathbf{b} = (\bm{b}_1, \allowbreak \bm{b}_2, \ldots, \bm{b}_n)$, where $\bm{b}_1$ is a best response of agent 1 to $\mathbf{b}_{-i} = (\bm{b}_2, \ldots, \bm{b}_n)$ according to her valuation function $v_1$. Let $(A_1, \dots,A_n)$ be the output of Round-Robin$(\mathbf{b})$. 
In order to apply Lemma \ref{lem:main_lemma}, we need $v_1$ to induce a strict preference ranking over the goods. 
	
Instead, we are going to use Lemma \ref{lem:ties} first to get a valuation function $v'_1$ having all properties stated therein. Note that by the second bullet of Lemma \ref{lem:ties},  $\bm{b}_1$ is still a best response of agent 1 to $\mathbf{b}_{-i}$ in the instance $\mathcal{I}'=(N,M,(v'_1, \mathbf{v}_{-1}))$. So, we  apply Lemma \ref{lem:main_lemma} here.
That is, we  consider the hypothetical scenario implied by the lemma: keeping agents 2 through $n$ fixed, suppose that the valuation function of agent 1 is the function $v_1^*$ given by the lemma, and her bid $\bm{b}_1^*$ is the truthful bid for $v_1^*$. The first part of Lemma \ref{lem:main_lemma} guarantees that the output of Round-Robin$(\bm{b}_1^*, \mathbf{b}_{-i})$ remains $(A_1,\dots,A_n)$.

According to Lemma \ref{lem:ef1_of_RR}, no matter what others bid, if agent 1 (the agent with the highest priority here) reports her true values (i.e., according to $v^*_1$) to Round-Robin, the resulting allocation is \ef from her perspective.
	In our hypothetical scenario this translates into having $v_1^*(A_1)\ge v_1^*(A_i)$ for all $i \in N$. 
	Then the second and third parts of Lemma \ref{lem:main_lemma} imply that $v'_1(A_1)\ge v'_1(A_i)$ for all $i \in N$.

	Suppose for a contradiction that there is a $j\in N$, such that $v_1(A_1)< v_1(A_j)$. By the definition of $\varepsilon$ in the statement of Lemma \ref{lem:ties} we have $v_1(A_1) \le v_1(A_j) - \varepsilon$. This implies \[v'_1(A_1) \le v_1(A_1) + \varepsilon / 3 \le v_1(A_j) - 2\varepsilon / 3 \le v'_1(A_j) - 2\varepsilon / 3 < v'_1(A_j)\,,\] which contradicts  $v'_1(A_1)\ge v'_1(A_j)$ that we showed above. We conclude that agent 1 does not envy (with respect to $v_1$ any bundle in the original instance.
\hfill\qed

\bibliographystyle{plainnat}
\bibliography{journalRef.bib}

\end{document}